\documentclass[conference,a4paper]{IEEEtran}
\IEEEoverridecommandlockouts

\usepackage{cite}
\usepackage{amsmath,amssymb,amsfonts}
\usepackage{amsthm}
\newtheorem{theorem}{Theorem}

\usepackage{graphicx}
\usepackage[small]{caption}

\usepackage{subcaption}
\usepackage{adjustbox}
\usepackage{textcomp}
\usepackage{multirow}
\usepackage{multicol}
\usepackage{booktabs}
\usepackage{bigstrut}
\usepackage{verbatim}
\usepackage{stackengine}
\usepackage{array}
\usepackage{rotating}
\usepackage{xcolor}
\usepackage{dblfloatfix}
\usepackage{booktabs}
\usepackage{multirow}
\usepackage{makecell}
\usepackage{adjustbox}
\usepackage{algorithm}
\usepackage{algpseudocode} 
\usepackage[table,xcdraw]{xcolor}
\definecolor{mygreen}{RGB}{0,130,0}

\usepackage[a4paper, total={184mm,239mm}]{geometry}
\AtBeginDocument{
  \providecommand\BibTeX{{
    \normalfont B\kern-0.5em{\scshape i\kern-0.25em b}\kern-0.8em\TeX}}}

\title{RIVERPlace: 
\underline{R}epairing \underline{I}nterconnect \underline{V}iolations with 
\underline{E}fficient 
\underline{R}etiming and Incremental \underline{Place}ment for AQFP Circuits}

\author{
\IEEEauthorblockN{Robert S. Aviles\IEEEauthorrefmark{1},
Ziyu Liu\IEEEauthorrefmark{1},
Sasan Razmkhah,
Massoud Pedram,
Peter A. Beerel}
    
\IEEEauthorblockA{Department of Electrical and Computer Engineering\\
University of Southern California, Los Angeles, CA 90007 USA\\
Email: \{rsaviles, zliu4130, razmkhah, pedram, pabeerel\}@usc.edu}

\IEEEauthorblockA{\IEEEauthorrefmark{1}These authors contributed equally to this work and are considered co–first authors.}}

\begin{document}

\maketitle

\begin{abstract}
The Adiabatic Quantum-Flux-Parametron (AQFP) offers near-Landauer-limit energy efficiency but faces significant scalability challenges due to strict path balancing and limited drive strength. To address this, we propose RIVERPlace, a framework that integrates long-wire pipelining, retiming, and incremental placement to resolve interconnect violations with minimal disruption. RIVERPlace first applies placement-aware retiming to repair violations without increasing logical depth. When depth increases are necessary, we introduce Buffer Cut Insertion (BCI), which formulates violation resolution as a constrained global edge-selection problem reducible to a maximum topological cut, thereby enabling an exact polynomial-time solution. By selectively pipelining edges across multiple rows, BCI avoids excessive buffer insertion while resolving interconnect violations. Experimental results demonstrate that RIVERPlace consistently outperforms prior AQFP placement approaches, reducing placement overhead by more than an order of magnitude in inserted buffers, 3$\times$ in placement-induced depth, and over 2$\times$ in circuit area, while also reducing runtime by more than an order of magnitude and latency by 38\%. These improvements enable the first post-routing, timing-closed implementations of the complete open-source AQFP benchmark suite, including larger circuits like \texttt{alu32}.

\end{abstract}

\begin{IEEEkeywords}
Physical Design, Retiming, Design Automation, Digital Circuits
\end{IEEEkeywords}

\section{Introduction}

The Adiabatic Quantum Flux Parametron (AQFP) is an emerging superconducting logic family that offers near-Landauer-limit energy efficiency, achieving up to 100× improvement in energy–delay product (EDP) over CMOS even with cryogenic cooling~\cite{cooling_overhead}. This makes AQFP attractive for applications such as quantum control~\cite{takeuchi2024microwave}, in-fridge data preprocessing~\cite{IcePack}, neural networks~\cite{BNN_AQFP2023}, cryptographic accelerators~\cite{SCE-NTT}, stochastic computing~\cite{DL_AQFP2019}, and energy-efficient data centers~\cite{MANA}.

Despite its advantages, AQFP introduces unique physical design challenges. In particular, AQFP circuits exhibit limited drive strength, large cell footprints, and gate-level clocking requirements. The combination of weak drive strength and strict row-wise placement imposed by the AC clock distribution results in numerous interconnect violations.

Previous work addresses these violations through buffer insertion. However, because every AQFP gate, including buffers, is clocked, \textit{path balancing} requires all paths to a given node to traverse the same number of logic stages. Consequently, inserting a buffer on a violating interconnect requires compensating buffers on all parallel paths, leading prior placement methods to insert an entire buffer row between adjacent placement rows~\cite{TAAS,DLPlace,GORDIAN,superflow}. As illustrated in Fig.~\ref{fig:RIVER_BRI}, this Buffer Row Insertion (BRI) approach introduces substantial area and circuit-depth overhead.

\begin{figure}[t] 
\begin{subfigure}[t]{0.25\textwidth}
\centering
\includegraphics[width=\linewidth]{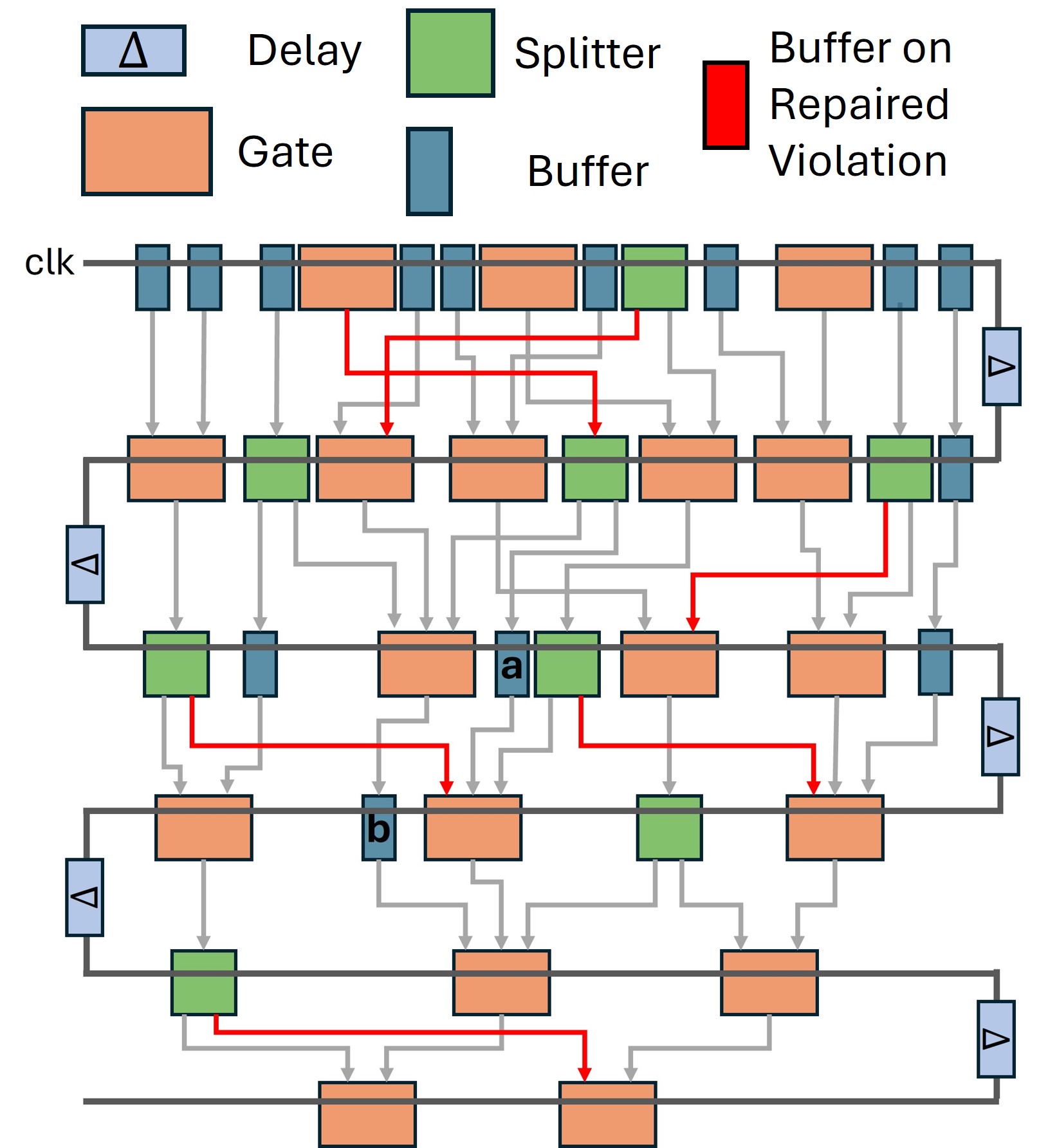}
\caption{AQFP circuit with interconnect\\violations in RED}\label{fig:RIVER_base}
\end{subfigure}%
\begin{subfigure}[t]{0.25\textwidth}
\centering
\includegraphics[width=\linewidth]{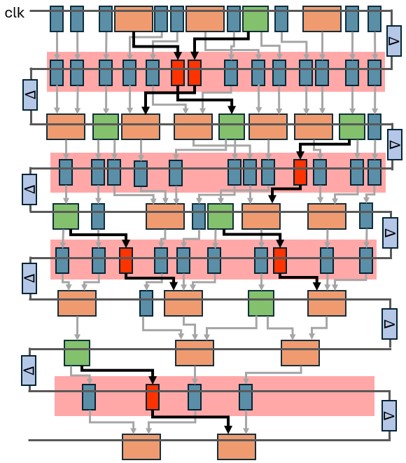}
\caption{Solution after traditional Buffer Row Insertion (red rows), increasing logical depth, area, and latency}\label{fig:RIVER_BRI}
\end{subfigure}
\caption{Illustrative impact of AQFP Placement Overhead}
\vspace{-2em}
\end{figure}

\begin{figure*}[t] 
\begin{subfigure}[t]{0.32\textwidth}
\centering
\includegraphics[width=\linewidth]{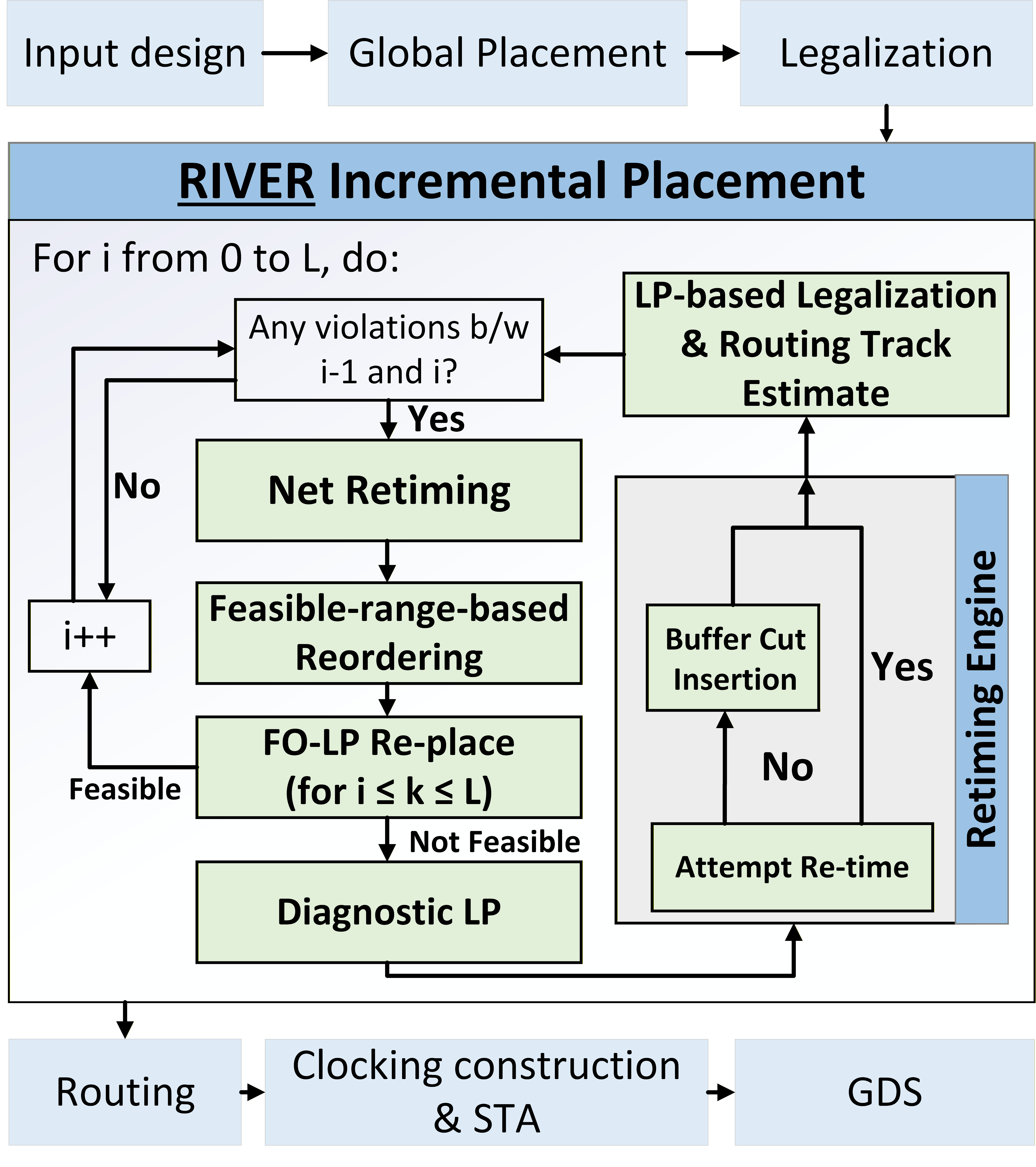}
\caption{RIVERPlace Design Flow}\label{fig:Flow}
\end{subfigure}\hfill
\begin{subfigure}[t]{0.32\textwidth}
\centering
\includegraphics[width=\linewidth]{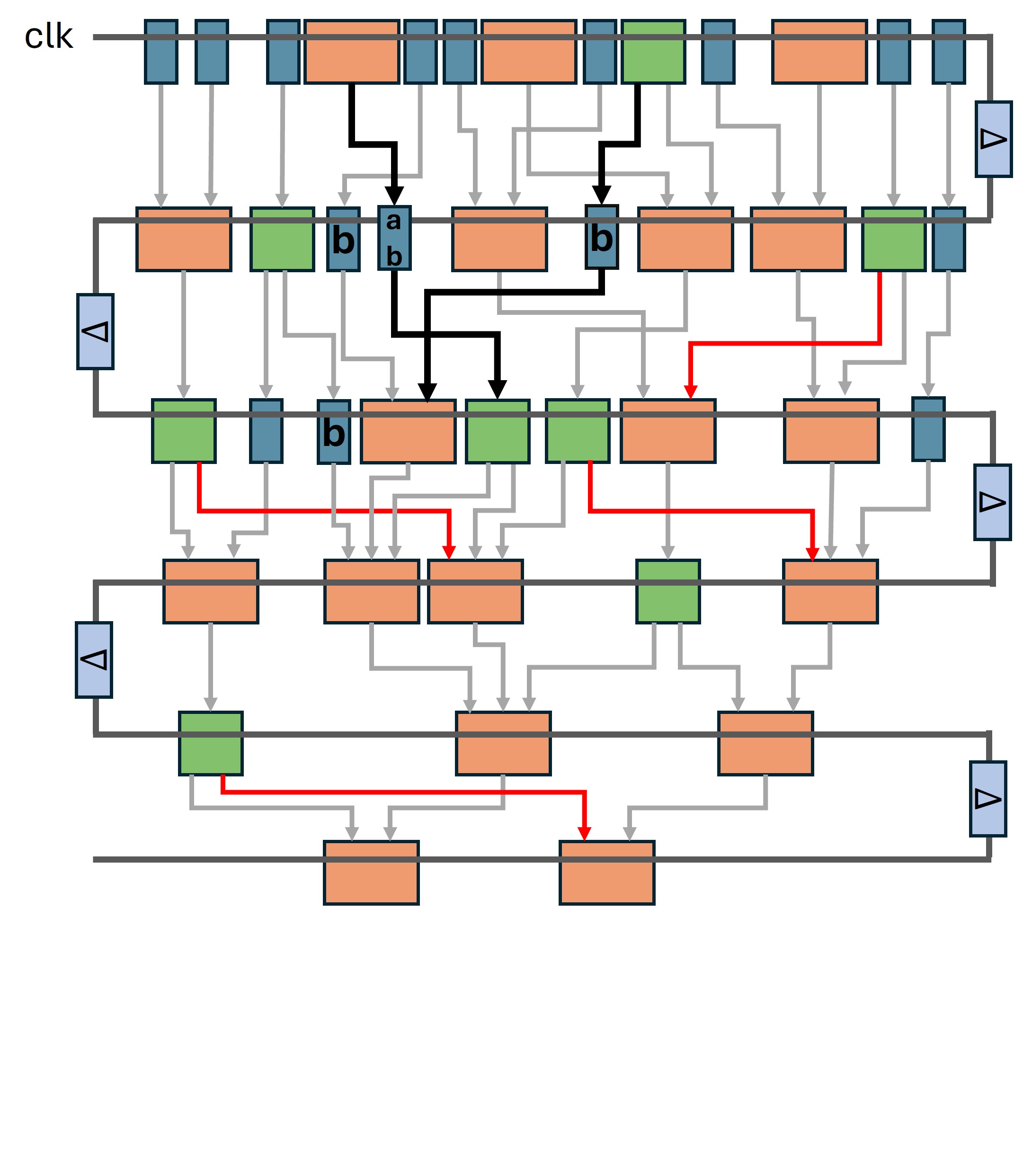}
\caption{Retiming solution to resolve row 1 violations using buffers a \& b from Fig.~\ref{fig:RIVER_base} without increasing circuit depth.  
}\label{fig:TFO}
\end{subfigure}\hfill
\begin{subfigure}[t]{0.32\textwidth}
\centering
\includegraphics[width=\linewidth]{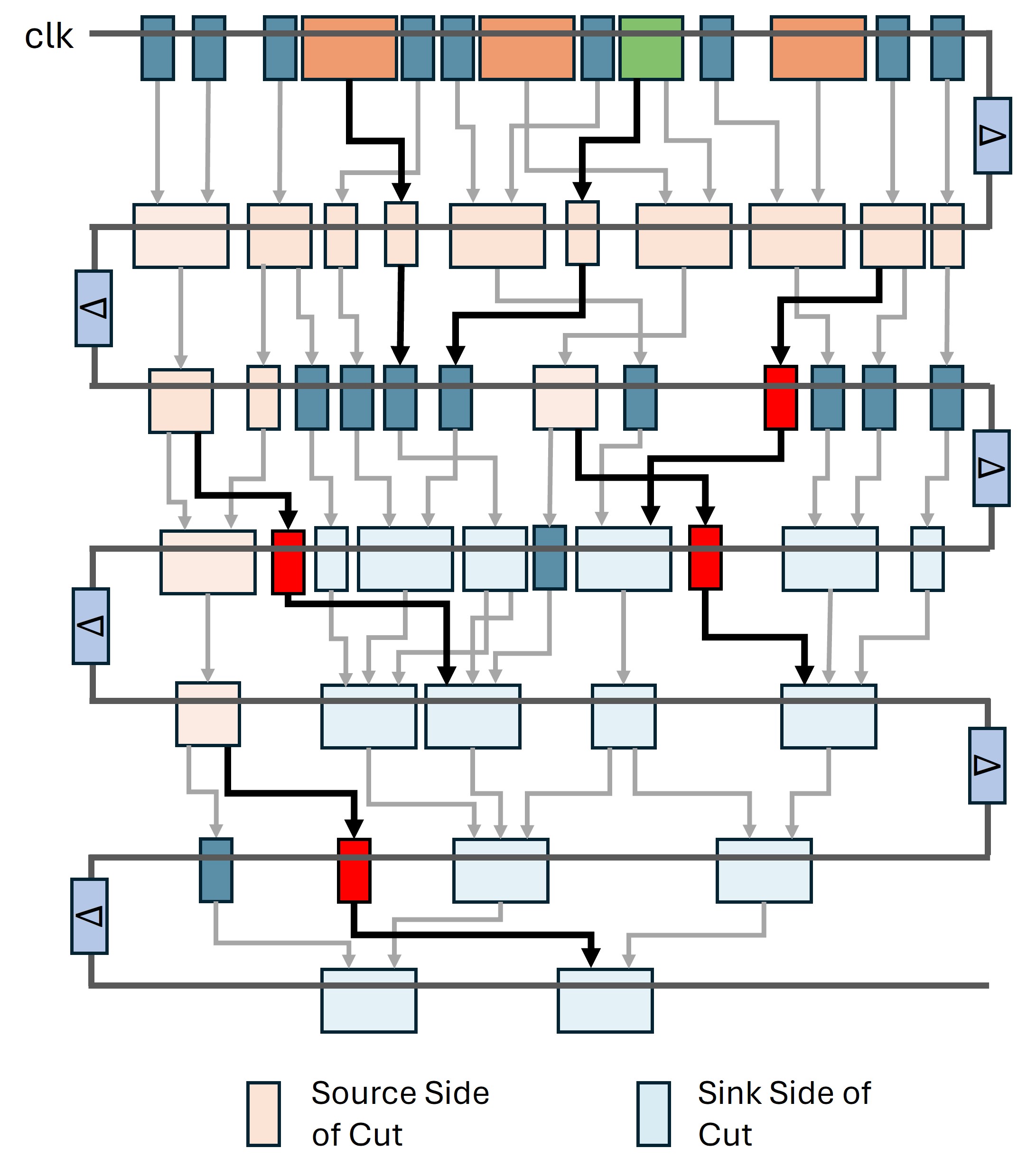}
\caption{Buffer CUT insertion solution of violations in Fig.~\ref{fig:TFO} that increases circuit depth but has reduced overhead compared to Fig.~\ref{fig:RIVER_BRI} }\label{fig:BCI}
\end{subfigure}\hfill
\caption{RIVERPlace Approach to Reducing AQFP Placement Overhead}
\vspace{-1em}
\end{figure*}

Although AQFP interconnect constraints are primarily driven by feasibility rather than performance, they share similarities with the challenges of long-wire pipelining in CMOS design. Thus, motivated by previous CMOS techniques that use retiming to manage interconnect delays~\cite{RetimingPipelining,PhysicalRetiming}, we propose integrating long-wire pipelining with retiming and placement. 

In AQFP, inserting a buffer changes the logic levels of all downstream fanouts, which requires their placement to shift across rows. As a result, inserting buffers can lead to cascading structural changes, in which previously resolved violations reappear, and new violations are introduced.

To address this, RIVERPlace adopts an iterative, placement-aware approach that resolves interconnect violations row-by-row while selectively pipelining edges through global optimization, as illustrated in Fig.~\ref{fig:Flow}. Our method first attempts to resolve violations within the current row using placement-aware retiming of existing non-critical buffers, avoiding increases in circuit depth when possible (Fig.~\ref{fig:TFO}).

When retiming alone is insufficient, we apply a novel Buffer Cut Insertion (BCI) that frames the pipelining of violating interconnects as a global edge selection problem. Specifically, BCI identifies a topological cut of the circuit that maximizes the pipelining of violating edges, subject to path-balancing constraints. This enables the simultaneous resolution of violations across 
multiple rows, significantly reducing pipelining overhead compared to buffer row insertion (Fig.~\ref{fig:BCI}).

More concretely, RIVERPlace is the first AQFP framework to tightly integrate pipelining, retiming, and incremental placement. Our contributions are as follows:
\begin{itemize}
    \item A placement-aware retiming strategy that selectively retimes existing buffers to repair violations without increasing circuit depth, while avoiding the introduction of new downstream violations.
    \item A buffer cut insertion (BCI) technique that formulates violation resolution under path balancing constraints as a constrained global edge selection problem, reducible to a maximum topological cut and enabling an exact polynomial-time solution.
    \item An LP-based incremental detailed placement framework that iteratively resolves interconnect violations in coordination with pipelining and retiming.
    \item Experimental results demonstrating consistent improvements over prior AQFP placement approaches, including more than an order-of-magnitude reduction in JJ insertion overhead and runtime, a 3$\times$ reduction in placement-induced circuit depth, and a 38\% reduction in latency. These improvements enable the first reported post-routing, timing-closed implementations of the complete open-source AQFP benchmark suite.
\end{itemize}

\section{Background}

We begin by reviewing AQFP logic, which is the focus of this work, and the maximum cut and minimum cut problems that are 
critically related to our approach.

\subsection{AQFP Logic}

Adiabatic Quantum Flux Parametron (AQFP) is a superconducting, majority-based logic family~\cite{scl} that encodes binary information using the direction of current flow. Logic evaluation is driven by an AC excitation current that simultaneously provides power and clocking, enabling adiabatic operation with zero static power dissipation~\cite{zero_static}. AQFP circuits operate at GHz frequencies (typically around 5\,GHz) while dissipating approximately $10^{-20}$ J per logic gate~\cite{AQFP_Dynamic_Energy}.

The fundamental AQFP buffer consists of two superconducting loops as shown in Fig.~\ref{fig:AQFPBuffer}. During the rising edge of the AC excitation current, the polarity of the input current determines which loop stores persistent magnetic flux, producing the corresponding output current polarity while the excitation remains active~\cite{AQFP_Timing}. Because each logic gate has relatively low drive strength, signals with multiple fanouts must be driven through explicit clocked splitter cells, resulting in splitter trees that significantly increase logical depth and Josephson junction (JJ) count~\cite{heuristicASP,BS_TCAD_Lee}.

\begin{figure}[t]
\includegraphics[width=0.75\columnwidth]{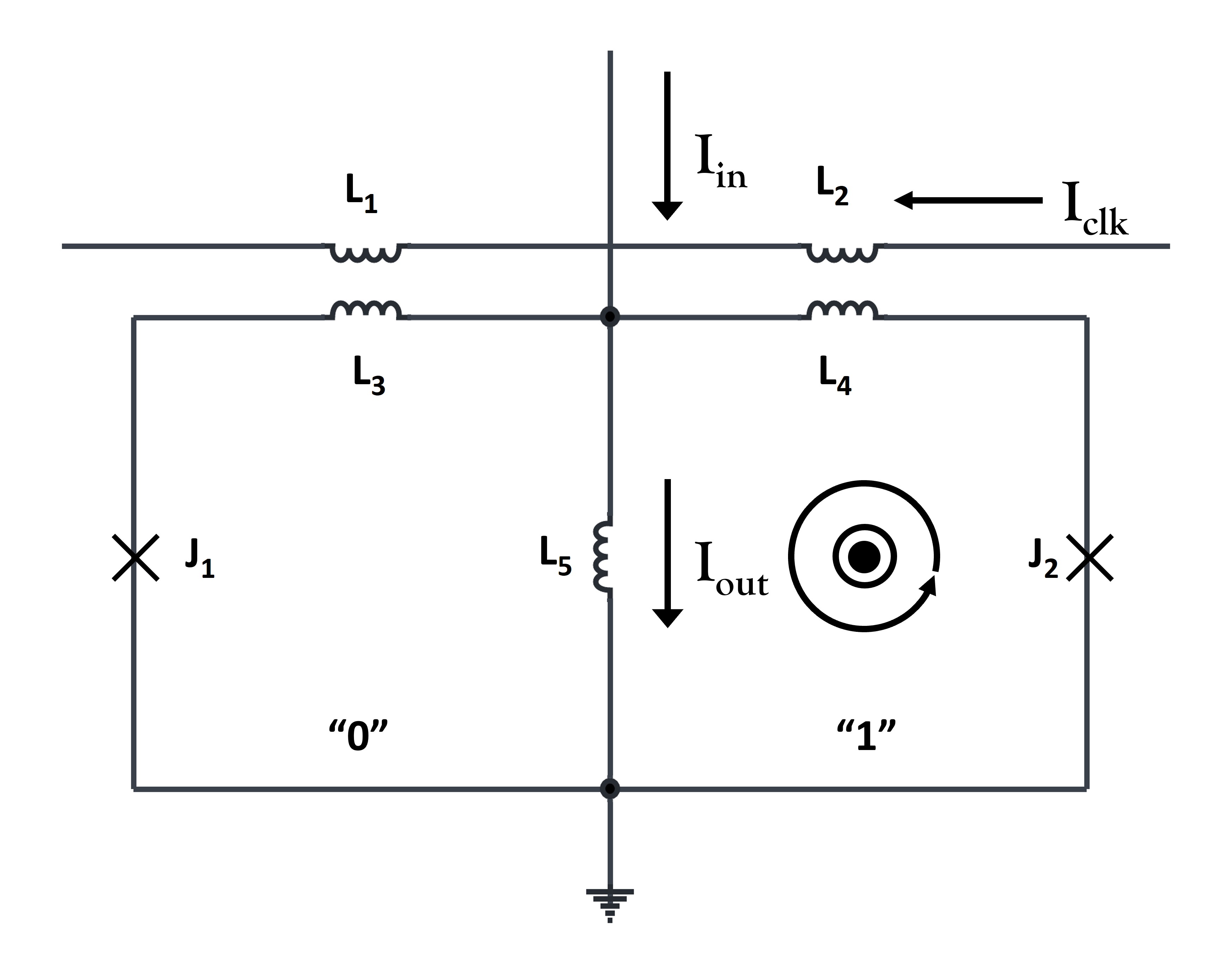}
\centering
\caption{Circuit construction of an AQFP Buffer. Current directions shown for when flux is in the right loop, corresponding to a logical `1'.}\label{fig:AQFPBuffer}
\end{figure}

\begin{figure}[htbp] 
    \centering
    \includegraphics[width=0.7\columnwidth]{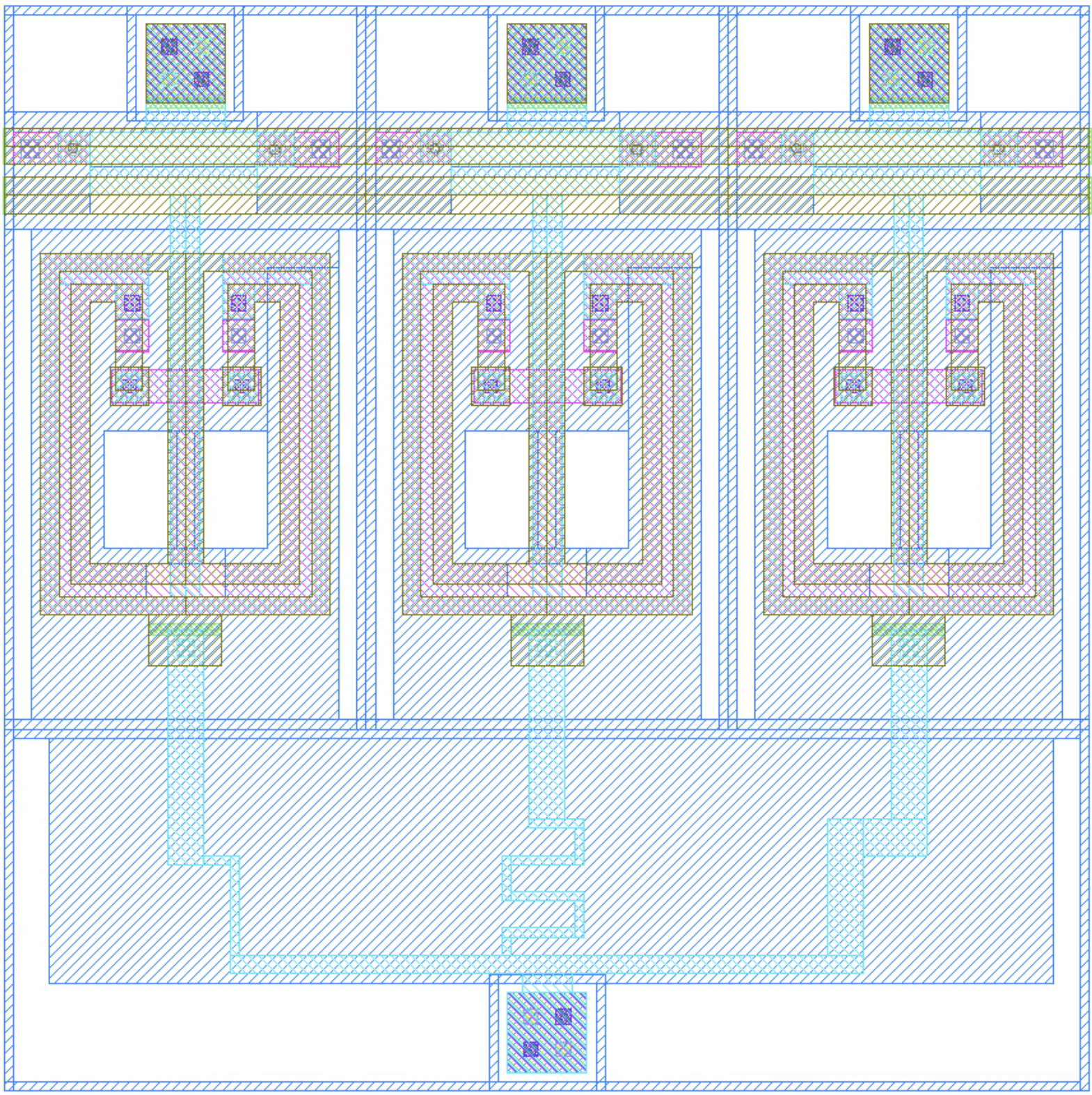}
    \caption{A sample AQFP gate layout, showing the outputs of 3 buffer cells merged to implement a majority-3 function.}
    \label{fig:Maj_gds}
\end{figure}

AQFP logic circuits are constructed from Majority-3 (MAJ3) gates together with a signal inversion gate, forming a functionally complete logic basis~\cite{maj3}. Figure~\ref{fig:Maj_gds} illustrates a representative MAJ3 layout, where the output current polarity is determined by the majority polarity of the three input currents driven by buffers. Practical AQFP cell libraries, therefore, consist of four fundamental cell types: positive polarity buffers, inverting buffers, constant value buffers, and branch cells that implement both majority and splitter functionality~\cite{scl}.

Automated AQFP design flows synthesize circuits using classical synthesis optimizations, before performing technology legalization, which inserts clocked splitter cells for multi-fanout signals and clocked buffers to satisfy path-balancing constraints. Consequently, technology legalization often dominates the final cell count, introducing substantial overhead~\cite{heuristicASP,BS_TCAD_Lee}. Although alternative clocking schemes relax path-balancing constraints and reduce some of this overhead~\cite{Nphaseclk,AvilesNPhase,phaseMatch,avilesPS_PA}, they generally increase interconnect distances and complicate physical design.

\subsection{AQFP Placement}

Following synthesis and technology legalization, AQFP circuits are placed using a row-based organization, in which all gates at the same logic level are assigned to a common placement row and share a clock signal. The clock is routed in a serpentine fashion across rows, as illustrated in Fig.~\ref{fig:RIVER_base}. Traditional AQFP designs employ four-phase AC clocking with adjacent rows offset by 90 degrees, while more recent approaches adopt delay-line clocking to reduce latency~\cite{DelayLineAQFP}.

GORDIAN~\cite{GORDIAN} introduced the first automated placement and routing framework for AQFP circuits, followed by timing-aware placement~\cite{TAAS} and the delay-line placement framework DLPlace~\cite{DLPlace}. Our work adopts delay-line clocking due to its improved latency characteristics.

Across these approaches, interconnect violations are resolved using \emph{Buffer Row Insertion} (BRI), in which an entire row of high-drive buffers is inserted to pipeline violating connections while preserving path balancing. Since AQFP buffers can drive substantially longer interconnects than standard logic gates, BRI simultaneously reduces wire lengths and provides increased drive strength. However, inserting complete buffer rows substantially increases the circuit depth, JJ count, and physical area, motivating alternative pipelining strategies.

\subsection{Max-Cut and Min-Cut problems}

\begin{figure}[t] 
\begin{subfigure}[t]{0.25\textwidth}
\centering
\includegraphics[width=\linewidth]{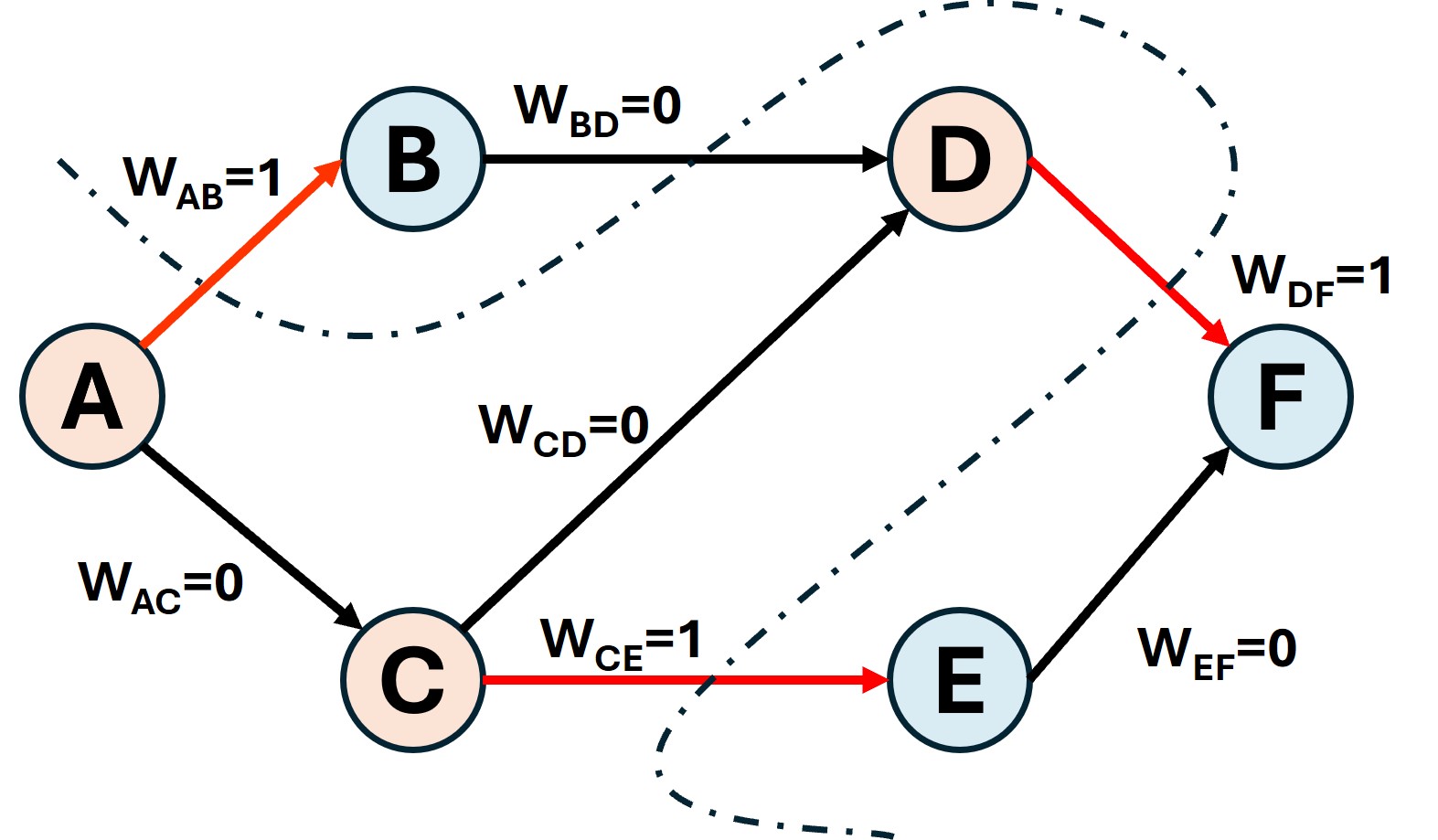}
\caption{Maximum Directed Cut}\label{fig:MaxDiCut}
\end{subfigure}%
\begin{subfigure}[t]{0.25\textwidth}
\centering
\includegraphics[width=\linewidth]{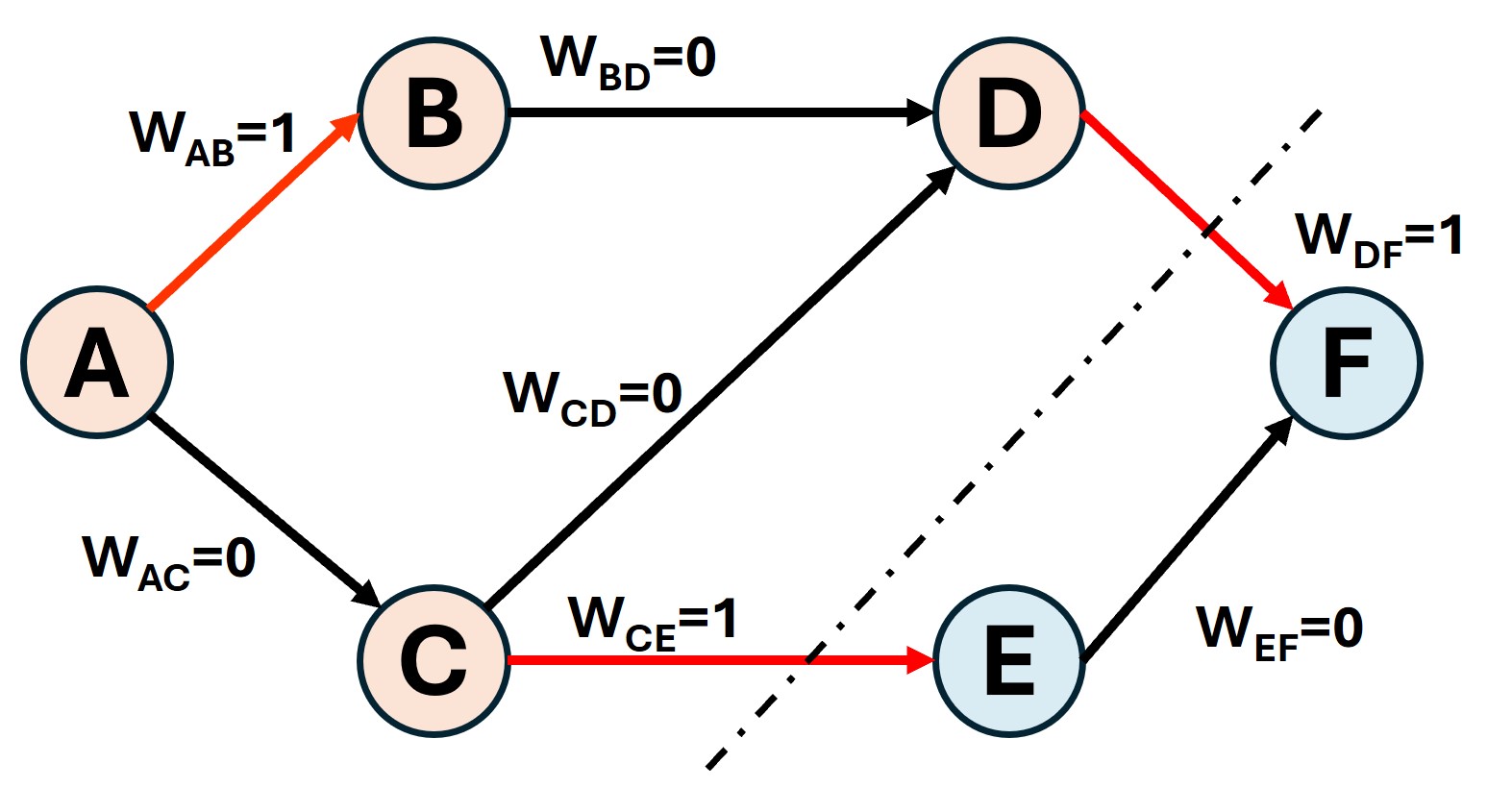}
\caption{Maximum Topological Cut}\label{fig:MaxTopCut}
\end{subfigure}
\caption{Maximum Cut Variants on a Directed Acyclic Graph}
\vspace{-1em}
\end{figure}

The maximum cut problem seeks to partition the vertex set of a graph $G=(V,E)$ into two disjoint sets $(S, V \setminus S)$ such that the total weight of edges crossing the partition is maximized. In directed graphs, this corresponds to selecting a subset $S \subseteq V$ that maximizes the weight of edges directed from $S$ to its complement:

\begin{equation}
\max_{S \subseteq V}
\sum_{(u,v)\in E,\; u \in S,\; v \notin S} w_{uv},
\label{eq:maxcut}
\end{equation}
as illustrated in Fig.~\ref{fig:MaxDiCut}. This problem is NP-complete and is one of Karp’s 21 classical NP-complete problems~\cite{karp1972reducibility}.

However, when additional \emph{precedence constraints} are imposed, the problem becomes tractable. Precedence constraints require that if a node $v \in S$, then all of its predecessors must also belong to $S$. $S$ is then considered \emph{predecessor-closed} and in directed acyclic graphs (DAGs), this yields a \emph{maximum topological cut}. As shown in Fig.~\ref{fig:MaxTopCut}, this constraint ensures that all edges crossing the cut are directed from $S$ to $V \setminus S$. Under these constraints, the problem can be solved exactly in polynomial time via a reduction to the minimum $s$-$t$ cut problem~\cite{MaxTopCut}.

The minimum $s$-$t$ cut problem is defined on a directed graph $G=(V,E)$ with nonnegative edge capacities $c_{uv}$. A flow $f:E \to \mathbb{R}_{\ge 0}$ assigns a value to each edge such that $0 \le f_{uv} \le c_{uv}$ and flow conservation holds at all nodes except the source $s$ and sink $t$:

\begin{equation}
\sum_{u:(u,v)\in E} f_{uv}
=
\sum_{w:(v,w)\in E} f_{vw},
\quad
\forall v \in V \setminus \{s,t\}.
\label{eq:flow_conservation}
\end{equation}
The maximum flow problem seeks a feasible flow that maximizes the total flow from $s$ to $t$. By the max-flow min-cut theorem~\cite{ford1956max}, the value of the maximum flow equals the capacity of the minimum $s$-$t$ cut.

Given a flow $f$, the corresponding residual graph $G_f$ contains forward edges with residual capacity $c_{uv} - f_{uv}$ and reverse edges with capacity $f_{uv}$. The set of vertices reachable from $s$ in $G_f$ defines the source-side partition $S$, and the edges from $S$ to $V \setminus S$ constitute a minimum $s$-$t$ cut.

In~\cite{MaxTopCut}, the maximum topological cut problem is reduced to a minimum $s$-$t$ cut by constructing an augmented graph with precisely derived baseline capacities by subtracting the original edge weights from an applied feasible flow:

\begin{equation}
c^+_{uv} = f_{uv} - w_{uv}.
\label{eq:residual_capacity}
\end{equation}
In this construction, the corresponding $s$-$t$ cut in the augmented graph with minimum capacity is set by the maximum total weight of edges crossing the cut in the original graph. Thus, the maximum topological cut equals the minimum $s$-$t$ cut of the transformed network.

The topological (precedence) constraint is enforced by ensuring all edges carry positive flow, which introduces reverse residual edges and guarantees that the reachable set $S$ is predecessor-closed.

Interested readers can refer to ~\cite{MaxTopCut} for more details on this formulation. Solving this maximum topological cut forms the foundation of our Buffer Cut Insertion (BCI) approach, which leverages topological constraints to maintain path balancing while efficiently selecting global pipelining.

\section{Incremental Placement}

Our incremental placement is applied after an initial global placement, obtained using a method similar to GORDIAN-based placement~\cite{GORDIAN}, as well as after any structural modifications introduced by the retiming engine (described in Section \ref{sec:retimingeng}). 

Let $L$ denote the total number of rows in the current placement. We consider a row-based placement of a circuit represented by the directed graph $G=(V,E)$, where each gate $i \in V$ is assigned to a row $r_i \in \{1,\ldots,L\}$ and $E$ denotes the set of interconnects between gates. Each gate $i$ is associated with an interconnect budget $\beta_i$, determined by its drive strength, which specifies the maximum physical length over which $i$ can reliably drive a sink gate~\cite{IC_limits_AQFP}. The interconnect budget of each cell type is calculated through SPICE-level simulation to ensure correct adiabatic switching of the sink gate under the target cryogenic operating conditions.

Let $V_r$ denote the set of gates in row $r$, $E_r \subseteq E$ the set of connections incident to row $r$, $W_r$ the row width, and $w_i$ representing gate $i$'s width. All gates in the same row share a common vertical coordinate $y_r$, and each gate is represented by its horizontal coordinate $x_i$ at the gate's center. 
The placement output is constrained to satisfy boundary and non-overlap constraints under a given order:
\begin{align}
\frac{w_i}{2} \le x_i \le W_r - \frac{w_i}{2}, \quad &\forall i \in V_r,
\label{eq:boundary} \\
x_j - x_i \ge \frac{w_i + w_j}{2}, \quad &\forall (i,j) \in V_r.
\label{eq:nonoverlap}
\end{align}

With $y$ locations fixed by row assignments, the placement objective is to assign each gate's $x$-location to satisfy interconnect length constraints; and when this is infeasible, resolve such violations in topological order before coordinating with the retiming engine to resolve any remaining violations. More precisely, for each connection $(u,v) \in E$, a violation occurs when the interconnect length exceeds the source gate's prescribed Manhattan budget $\beta_{u}$:

\begin{equation}
X_{uv} + Y_{uv} \le \beta_{u}.
\end{equation}
where $X_{uv}$ and $Y_{uv}$ denote its horizontal and vertical components, respectively. The vertical component $Y_{uv}$ is determined by the current row assignment and the channel width between each row, which may change across iterations (see Section~\ref{subsec:lp_legalization}). Therefore, the constraint on the horizontal component can be written as

\begin{equation}
X_{uv} \le \beta_{u} - Y_{uv},
\end{equation}
where the upper bound is evaluated using the current gate locations at each iteration.

Because jointly optimizing cell ordering and exact coordinates leads to a highly non-convex search space, we decouple the problem: we first apply local retiming and heuristic feasible-range reordering to establish a high-quality gate ordering with a reduced number of violations, and then fix this order to formulate the coordinate assignment as a convex Linear Program (LP), enabling optimal placement for the given sequence while systematically identifying irreconcilable violations via a secondary diagnostic LP, which are then passed to the retiming engine.

\subsection{Net Retiming}
As an initial step of incremental re-placement, a greedy local net retiming procedure is attempted to mitigate interconnect violations. This localized retiming is separate from the global retiming operations explored during the retiming engine, as it only considers retiming an immediate fanout of a violating connection. Violating connections are prioritized by the magnitude of violation, and retiming is attempted for each candidate in descending order.

For each candidate connection, a local retiming is considered for a connection of the form $u \rightarrow b \rightarrow v$, where $u$ denotes a single-fanout driving gate, $b$ denotes a buffer, and $(b,v)$ lies within the allowable driving range of $u$, such that direct driving of $v$ by $u$ remains feasible after retiming $b$ across $u$.

\begin{equation}
X_{bv} + Y_{bv} \le \beta_{u},
\end{equation}
For each eligible candidate, a tentative swap between $u$ and the buffer is performed, followed by a legalization step to ensure that placement is feasible. This retiming is accepted only if the violation is resolved and no additional violations are introduced; otherwise, it is rejected.

\begin{figure}[t]
\centering
\includegraphics[width=\columnwidth]{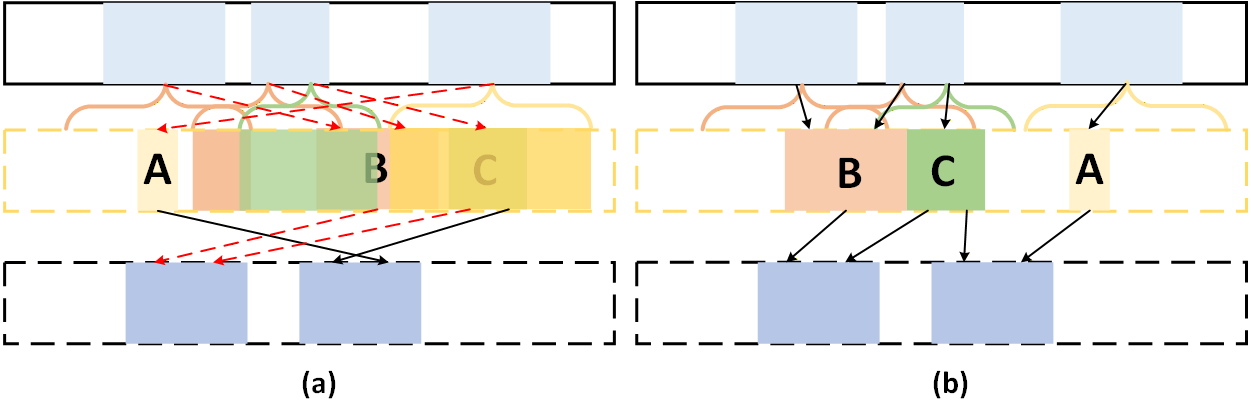}
\caption{Feasible-range-based reordering, with the fixed preceding row on top, the active row in the middle, and the unplaced following row at the bottom. 
(a) The feasible range of $g_A$ (yellow), $g_B$ (orange),
and $g_C$ (green) is the intersection of the brackets of the same color, and the red dashed lines mark the connections that violate $|x_i-x_j|>X_{ij}$. 
(b) After reordering, each gate lies inside its own feasible range, and no violating connection remains.}
\label{fig:feasible_range_reorder}
\end{figure}

\subsection{Feasible-range-based Reordering}
\label{sec:reordering}

Our linear programming-based optimization algorithm is constrained by the fixed ordering of gates within each row. We thus optimize this order as a pre-processing step, as summarized in Algorithm~\ref{alg:reordering}.

In particular, we first classify rows according to retiming-induced changes. A row $r$ is marked as \emph{active} if it contains any modified, added, or removed instance; otherwise, it is considered \emph{inactive}. Since the retiming engine makes predictive pipeline decisions based on the current placement order, gates in active rows may adjust their relative ordering, whereas gates in inactive rows retain their relative order. This ensures that changes induced by the retiming engine and LP converge to a stable placement solution. 

For each gate $i$ in an active row, a feasible range $[\ell_i, u_i]$ is derived from the distance constraints imposed by its fanins, as shown in Fig.~\ref{fig:feasible_range_reorder}(a). 
Each bracket corresponds to one such constraint, brackets of the same color belong to the same gate, and their intersection, drawn as the shaded region of that color, gives the feasible range of that gate. 
Only fanin gates affect the feasible range, because the rows are processed in topological order, so the positions in the preceding row are already fixed and act as hard constraints, whereas the positions in the following row are still free and therefore constrain nothing at this point.
A projected position $\tilde{x}_i$ is then computed based on the feasible range of gate $i$ and the positions of its fanin and fanout neighbors $\mathcal{N}_i$.
The projected positions $\tilde{x}_i$ are used to induce an ordering of the gates and to provide reference positions for subsequent LP-based refinement. The gates are sorted according to $\tilde{x}_i$, with ties broken by the original ordering, and then a pass of neighborhood swap refinement is applied, which produces the order in Fig.~\ref{fig:feasible_range_reorder}(b).

\begin{algorithm}
\caption{Feasible-range-based Reordering}
\label{alg:reordering}
\begin{algorithmic}[1]
\Require Initial gate positions $x_i$, row assignment $V_r$, netlist $\mathcal{N}_i$
\Ensure Updated gate ordering in each row $V_r$, and reference positions $\tilde{x}_i$
\For{each active row $r$}
    \For{each gate $i \in V_r$}
        \State $[\ell_i, u_i] \leftarrow \left[\frac{w_i}{2},\ 
               W_r - \frac{w_i}{2}\right]$
        \For{each driver $d$ of gate $i$}
            \State $[\ell_i, u_i] \leftarrow [\ell_i, u_i] \cap 
                   [x_d - X_{di},\ x_d + X_{di}]$
        \EndFor
        \State $sx \leftarrow 0$, $tw \leftarrow 0$, 
        $\mathit{violation} \leftarrow \text{False}$
        \For{each neighbor $j \in \mathcal{N}_i$}
            \State $w_{ij} \leftarrow \max(1, |x_j - x_i| / X_{ij})$
            \State $sx \leftarrow sx + x_j \cdot w_{ij}$
            \State $tw \leftarrow tw + w_{ij}$
            \If{$w_{ij} > 1$}
                \State $\mathit{violation} \leftarrow \text{True}$
            \EndIf
        \EndFor
        \If{$\mathit{violation}$}
            \State $x_i^{\mathrm{ref}} \leftarrow sx / tw$
            \State $x_{\text{mid}} \leftarrow \frac{x_i^0 + x_i^{\mathrm{ref}}}{2}$
            \If{$\ell_i \le u_i$}
                \State $\tilde{x}_i \leftarrow \mathrm{clip}(x_{\text{mid}}, \ell_i, u_i)$
            \Else
                \State $\tilde{x}_i \leftarrow x_{\text{mid}}$
            \EndIf
        \Else
            \State $\tilde{x}_i \leftarrow x_i$
        \EndIf
    \EndFor
    \State Sort gates in $V_r$ by $\tilde{x}_i$, tie-break by original ordering
    \For{each adjacent pair $(i, j)$ in $V_r$}
        \If{swapping $i$ and $j$ reduces total violations}
            \State Swap the order and position of $i$ and $j$
        \EndIf
    \EndFor
\EndFor
\end{algorithmic}
\end{algorithm}

\subsection{Fixed-Order LP Re-placement}

Our fixed-order LP (FO-LP) re-placement is constrained to iteratively resolve inter-row violations, progressively converging towards a final placement solution by resolving them row by row in topological order.  To ensure the locations of gates in upstream rows are not disturbed, re-placement is performed over rows $r \le k \le L$, where $r$ is the current row under placement. 

Gate positions are optimized subject to the boundary and non-overlap constraints defined in~\eqref{eq:boundary}--\eqref{eq:nonoverlap}, with ordering fixed according to Section~\ref{sec:reordering}.

To prioritize violation resolution for connections between row $r$ and its predecessor row $r-1$, hard constraints are enforced:
\begin{align}
|x_u - x_v| \le X_{uv}.
\end{align}
While connections involving downstream rows $k > r$, have slack variables $\epsilon_{uv} \ge 0$ introduced to transform interconnect violations into an optimization parameter to minimize:
\begin{align}
|x_u - x_v| \le X_{uv} + \epsilon_{uv}.
\end{align}
Here, $X_{uv}$ denotes an upper bound on the allowable horizontal displacement, derived from the current placement-and-routing estimate, and it may vary across iterations.

All absolute value terms in the formulation are linearized using standard auxiliary variables to maintain a convex search space. Specifically, we introduce non-negative displacement variables $d_i \ge 0$ for each gate $i \in V_{\text{opt}}$ to represent the coordinate deviation $|x_i - \tilde{x}_i|$ from the reference position. The objective function in Equation~(15) is optimized subject to the following linear boundary constraints:
\begin{align}
    x_i - \tilde{x}_i &\le d_i, \quad \forall i \in V_{\text{opt}} \\
    \tilde{x}_i - x_i &\le d_i, \quad \forall i \in V_{\text{opt}}
\end{align}
Similarly, the absolute horizontal displacement constraints $|x_u - x_v| \le X_{uv} + \epsilon_{uv}$ for downstream connections ($r_u > r$) are linearized using twin inequalities:
\begin{align}
    x_u - x_v &\le X_{uv} + \epsilon_{uv}, \quad \forall (u,v) \in E, r_u > r \\
    x_v - x_u &\le X_{uv} + \epsilon_{uv}, \quad \forall (u,v) \in E, r_u > r
\end{align}
This formulation yields a standard, highly scalable linear program that can be solved optimally in polynomial time.

Subject to the hard interconnect constraints at the 
current frontier, the LP objective minimizes the slack violations in downstream connections and the displacement from the initial placement:
\begin{equation}
\min \;
\lambda_s \sum_{\substack{(u,v) \in E \\ r_u > r}} \epsilon_{uv}
+ \lambda_m \sum_{i \in V_{\mathrm{opt}}} |x_i - \tilde{x}_i|,
\end{equation}
where $\tilde{x}_i$ denotes the reference position obtained from the reordering step (Algorithm~\ref{alg:reordering}), and $V_{\mathrm{opt}} = \bigcup_{k=r}^{L} V_k$ denotes the set of gates included in the optimization. 

\subsection{Diagnostic LP}

Although the FO-LP enforces interconnect constraints at the current frontier as hard constraints, these constraints may not always be satisfiable. Rather than uniformly treating all incident edges as violating,  we construct a diagnostic LP restricted to row $r$, in which all interconnect constraints are relaxed using slack variables $\epsilon_{uv} \ge 0$:
\begin{align}
|x_u - x_v| \le X_{uv} + \epsilon_{uv}, \quad \forall (u,v) \in E_r,
\end{align}

The LP minimizes the total violation:
\begin{equation}
\min \sum_{(u,v)\in E_r} \epsilon_{uv}
\end{equation}

The resulting slack values identify which specific edges cannot be resolved through placement alone and guide the retiming engine in restoring feasibility.

\subsection{LP-Based Legalization and Track Estimation}
\label{subsec:lp_legalization}
After pipelining and retiming, we perform fixed-order LP-based legalization to obtain a legalized placement that satisfies all boundary and non-overlap constraints, serving as a feasible starting point for the next iteration.

Following legalization, routing demand between adjacent rows is estimated using an interval-density model: each inter-row connection is projected onto the horizontal axis, retaining only its x-component, and the maximum interval overlap is computed to estimate the required number of routing tracks. The estimated track count is used to update the row spacing, which in turn determines $Y_{uv}$ for each connection, ensuring placement solutions accurately reflect routing channel widths.

\section{Retiming and Buffer Cut Insertion}\label{sec:retimingeng}
If placement alone cannot resolve interconnect violations between row $r-1$ and row $r$, we fix the violating paths while minimizing overhead by combining retiming and pipelining. 

Let $C$ be the set of violating edges originating from row $r-1$ that must be resolved, and let $A$ be the set of non-violating edges from the same row. Specifically, the diagnostic LP identifies $C$ as the edges with $\epsilon_{uv} > 0$ and $A$ as edges with $\epsilon_{uv} =0$. Moreover, to prioritize resolving violations in near rows, we let $D$ specify the maximum number of rows beyond $r$ that may be explored during retiming or buffer insertion.

To incorporate placement awareness, we are given a set $B \subseteq V$ of retimable buffers. Additionally, edges are weighted to guide optimization, with a weight of $1$ assigned to violating edges and $0$ otherwise. A buffer $j$ is considered retimable if its source gate $i$ can be directly connected to its sink $k$ without introducing a new violation when $j$ is moved. Formally, we define: 
\begin{equation}
\mathrm{Retimable}(j) =
\begin{cases}
1, & \text{if } |i_x - k_x| + |j_y - k_y| \le \beta_i \\
0, & \text{otherwise}
\end{cases}
\end{equation}

We first attempt to repair violations without increasing circuit depth by exploring the transitive fanout of each edge in $C$. Specifically, we search for a cut within each fanout cone consisting solely of retimable buffers. If such cuts exist for all edges in $C$, the corresponding buffers are retimed to level $r$, resolving violations without increasing logical depth.

If no feasible set of cuts exists, increasing circuit depth becomes necessary. In this case, we formulate violation resolution as a constrained maximum-topological-cut problem on the circuit graph. Rather than inserting buffers greedily or uniformly as a single row, we seek to select a maximum subset of violating edges on which buffers will be placed such that all edges in $C$ are pipelined while global path balancing constraints are preserved. 

The overall retiming procedure is summarized in Algorithm~\ref{alg:retime_engine}, while Sections~\ref{sec:tfo} and~\ref{sec:bci} describe the transitive fanout cut search and BCI formulations in detail.

\begin{algorithm}[t]
\caption{RETIMING\_ENGINE$(G, C, A, B,r,D)$}
\label{alg:retime_engine}
\begin{algorithmic}[1]
\Statex \textbf{Input:}
\Statex $G=(V,E)$: directed acyclic graph representation of the circuit
\Statex $C$: set of violating edges originating from level $r-1$
\Statex $A$: set of non-violating edges at level $r-1$
\Statex $B$: set of retimable buffers
\Statex $r$: source logic level under consideration
\Statex $D$: maximum circuit depth explored during retiming or buffer insertion
\Statex \textbf{Output:} $G'$: path-balanced design with all edges in $C$ pipelined

\State $\mathrm{\textit{buffer\_list}} \gets \emptyset$
\State $\mathrm{\textit{pass}} \gets \textbf{true}$

\ForAll{$e \in C$}
    \ForAll{$v \in \textit{sink(e)}$}
        \State $(\mathrm{\textit{pass}}, \mathrm{\textit{bfrs}}) \gets \mathrm{TFO\_CUT}(G,v, B,r,D)$ \Comment{Alg.~\ref{alg:tfo_cut}}
        \If{$\neg \mathrm{\textit{pass}}$}
            \State \textbf{break}
        \Else
            \State $\mathrm{\textit{buffer\_list}} \gets \mathrm{\textit{buffer\_list}} \cup \mathrm{\textit{bfrs}}$
        \EndIf
    \EndFor
    \If{$\neg \mathrm{\textit{pass}}$}
        \State \textbf{break}
    \EndIf
\EndFor

\If{$pass$}
    \State $G' \gets \mathrm{ApplyRetiming(G, \textit{buffer\_list}, r)}$
\Else
    \State $\mathrm{MaxTopCut} \gets \mathrm{BCI}(G,C, A,r,D)$ \Comment{Alg.~\ref{alg:bci}}
    \State $G' \gets \mathrm{InsertBuffers(G, \mathrm{MaxTopCut})}$
\EndIf

\State \Return $G'$

\end{algorithmic}
\end{algorithm}
\subsection{Retiming of Existing Buffers}\label{sec:tfo}

Retiming existing buffers to pipeline all edges in $C$ can be viewed as a variant of the classical retiming problem~\cite{leiserson1991retiming}, but with a simplified objective. Rather than minimizing global buffer count, our goal is to identify a feasible set of retimable buffers that can be propagated to all edges in $C$.

To maintain compatibility with our row-by-row placement framework, we restrict retiming to downstream buffers, within a bounded range. Specifically, we only consider buffers that can be retimed within $\Delta$ rows of the current level $r$, yielding a feasible region $[r, r+D]$. 
Furthermore, since row-based placement ensures that no edge in $C$ lies in the transitive fanout of another, retimable buffers can be shared across edges without conflict.

Under these conditions, the problem decomposes into independently identifying, for each edge $e_{uv}$ in $C$, a cut of retimable buffers within its transitive fanout cone. We implement this as a depth-first search (DFS) over the fanout graph, terminating either when the depth reaches $r+D$ or when a valid cut is identified. 

The procedure is formalized in Algorithm~\ref{alg:tfo_cut}, which operates on a single sink node $v$ of an edge in $C$ and returns a feasible cut within its transitive fanout cone, if one exists. These local solutions are then aggregated across all edges in $C$ by Algorithm~\ref{alg:retime_engine} to determine a globally feasible retiming. The overall complexity is $O\bigl(|C|(|V|+|E|)\bigr)$, with worst-case complexity $O\bigl(|E|(|V|+|E|)\bigr)$.

\begin{algorithm}[htbp]
\caption{TFO\_CUT$(G, v, B,r, D)$}
\label{alg:tfo_cut}
\begin{algorithmic}[1]
\Statex \textbf{Input:}
\Statex $G=(V,E)$: directed acyclic graph representation of the circuit
\Statex $v$: root node of the transitive fanout to explore
\Statex $B$: set of retimable buffers
\Statex $r$: source logic level under consideration
\Statex $D$: maximum depth explored for retiming
\Statex \textbf{Output:}
\Statex $pass$: feasibility indicator
\Statex $buffs$: set of buffer nodes to retime

\State $\mathrm{\textit{buffs}} \gets \emptyset$
\State $\mathrm{\textit{pass}} \gets \textbf{true}$

\ForAll{$n \in \mathrm{fo}(v)$}
    \If{$n \in B$}
        \State $\mathrm{\textit{buffs}} \gets \mathrm{\textit{buffs}} \cup \{n\}$
    \ElsIf{$\mathrm{level}(n) = r+D$}
        \State \Return $\textbf{false}, \emptyset$
    \Else
        \State $(\mathrm{\textit{pass}}, \mathrm{\textit{new\_buffs}}) \gets \mathrm{TFO\_CUT}(G, n, B,r, D)$
        \If{$\neg \mathrm{\textit{pass}}$}
            \State \Return $\textbf{false}, \emptyset$
        \Else
            \State $\mathrm{\textit{buffs}} \gets \mathrm{\textit{buffs}} \cup \mathrm{\textit{new\_buffs}}$
        \EndIf
    \EndIf
\EndFor

\State \Return $\textbf{true}, \mathrm{\textit{buffs}}$

\end{algorithmic}
\end{algorithm}

\subsection{Buffer Cut Insertion}\label{sec:bci}

When retiming of existing buffers is insufficient to resolve all violations in $C$, we apply \emph{Buffer Cut Insertion (BCI)}. The objective is to select a set of edges on which buffers will be inserted, such that all edges in $C$ are included while maximizing the number of additional violating edges pipelined up to depth $r+D$.

This formulation casts violation resolution as a global edge-selection problem subject to structural and path-balancing constraints. Under path-balancing constraints, this problem reduces to a constrained maximum-weight topological cut, thereby enabling an exact polynomial-time solution.

A topological cut $(S, T)$, where $S, T \subseteq V$, $S \cap T = \emptyset$, and $S \cup T = V$, partitions the graph such that no edge in $E$ is directed from a node in $T$ to a node in $S$. Equivalently, $S$ is predecessor-closed and $T = V \setminus S$ is successor-closed. In the context of BCI, nodes in $S$ retain their logic level, while nodes in $T$ experience a uniform level increase due to buffer insertion.

\begin{theorem}[Path-Balancing Preservation under Topological Cut]
Let $(S,T)$ be a topological cut of a path-balanced DAG $G=(V,E)$,
with $\mathrm{PI}\subseteq S$. If a buffer is inserted on every cut edge
\[
E_{S,T}=\{(u,v)\in E \mid u\in S,\;v\in T\},
\]
then every path from a primary input to any $v\in T$ receives exactly
one additional buffer and path balancing is preserved.
\end{theorem}

\begin{proof}
Suppose, for contradiction, that path balancing is violated after
buffer cut insertion. Then there exists some $v\in T$ with two
primary-input-to-$v$ paths that receive different numbers of inserted
buffers.

Let $P=(u_1,\ldots,u_m)$ be any directed path from a primary input to
$v$, and let
\[
E_P=\{(u_i,u_{i+1}) \mid 1\le i<m\}.
\]
Consider any such path
$P$, where $u_1\in\mathrm{PI}\subseteq S$ and
$u_m=v\in T$. 
Since $u_1\in S$ and $u_m=v\in T$, the path must contain at least one
edge in $E_{S,T}$. Suppose it contains two such edges. After its first
transition from $S$ to $T$, it must return to $S$ before crossing from
$S$ to $T$ again. This requires some edge $(u_j,u_{j+1})\in E$ with
\[
u_j\in T,\qquad u_{j+1}\in S,
\]
which contradicts the definition of a topological cut, since no edge
may be directed from $T$ to $S$. Therefore,
\[
|E_P\cap E_{S,T}|=1
\]
for every primary-input-to-$v$ path.

Since the original circuit is path balanced, for any node $v\in T$,
all primary-input-to-$v$ paths have equal logical depth. That is, for
any two such paths $P_1$ and $P_2$,
\[
|P_1| = |P_2|.
\]
As established above, every primary-input-to-$v$ path crosses
$E_{S,T}$ exactly once and therefore receives exactly one additional
buffer. Hence, after buffer insertion,
\[
|P_1'| = |P_1|+1
        = |P_2|+1
        = |P_2'|.
\]
Thus, all paths terminating at $v$ remain equal in logical depth.
Since this argument holds for every $v\in T$, while nodes in $S$
remain unchanged, the path-balancing condition is preserved
throughout the circuit, contradicting the assumption that buffer cut
insertion introduces a path imbalance.
\end{proof}

\vspace{0.4em}
\begin{figure}[t]
\centering

\begin{subfigure}[t]{0.48\columnwidth}
\centering
\includegraphics[width=\linewidth]{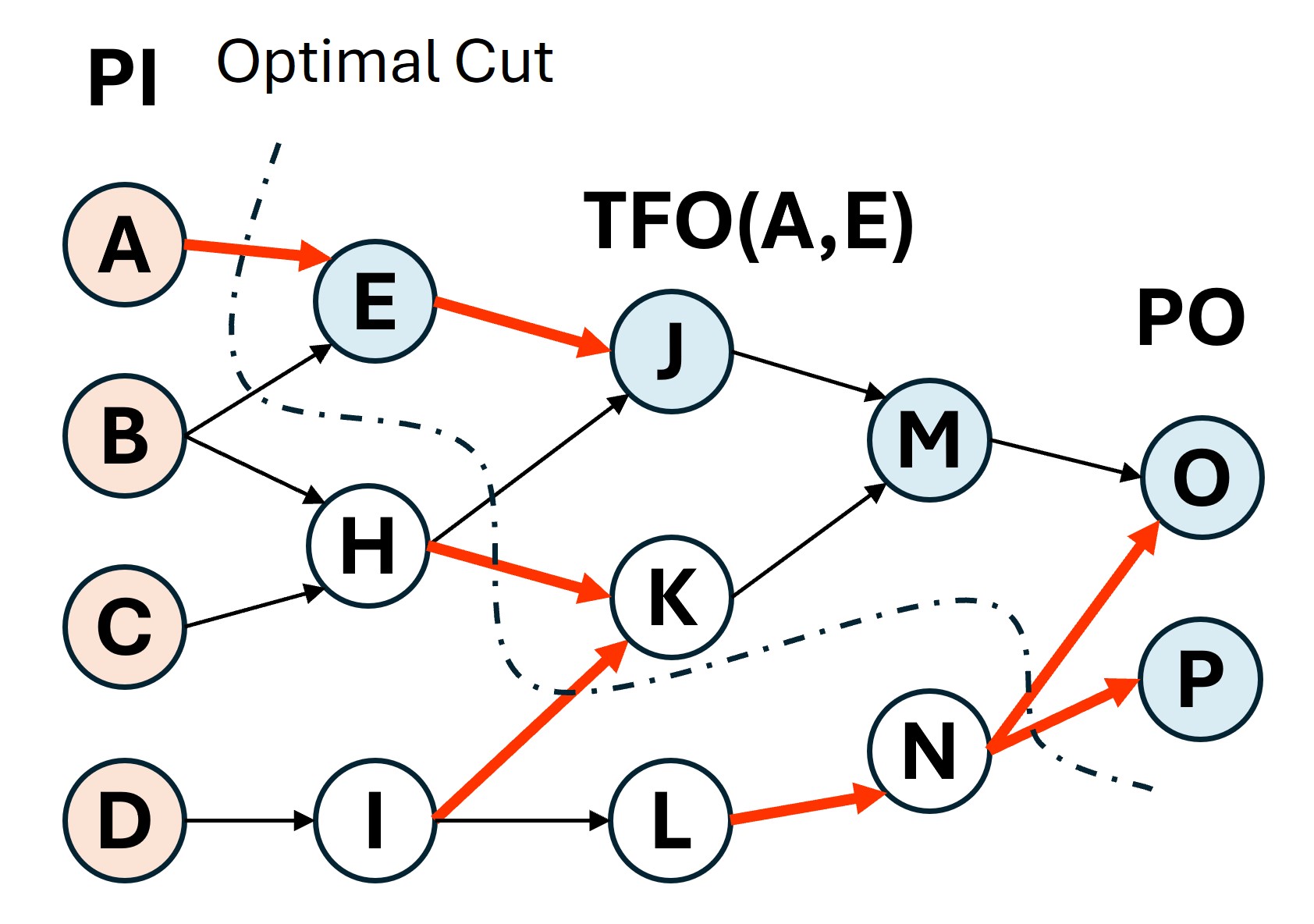}
\caption{Original graph $G$}
\label{fig:MaxTopCutG}
\end{subfigure}\hfill
\begin{subfigure}[t]{0.48\columnwidth}
\centering
\includegraphics[width=\linewidth]{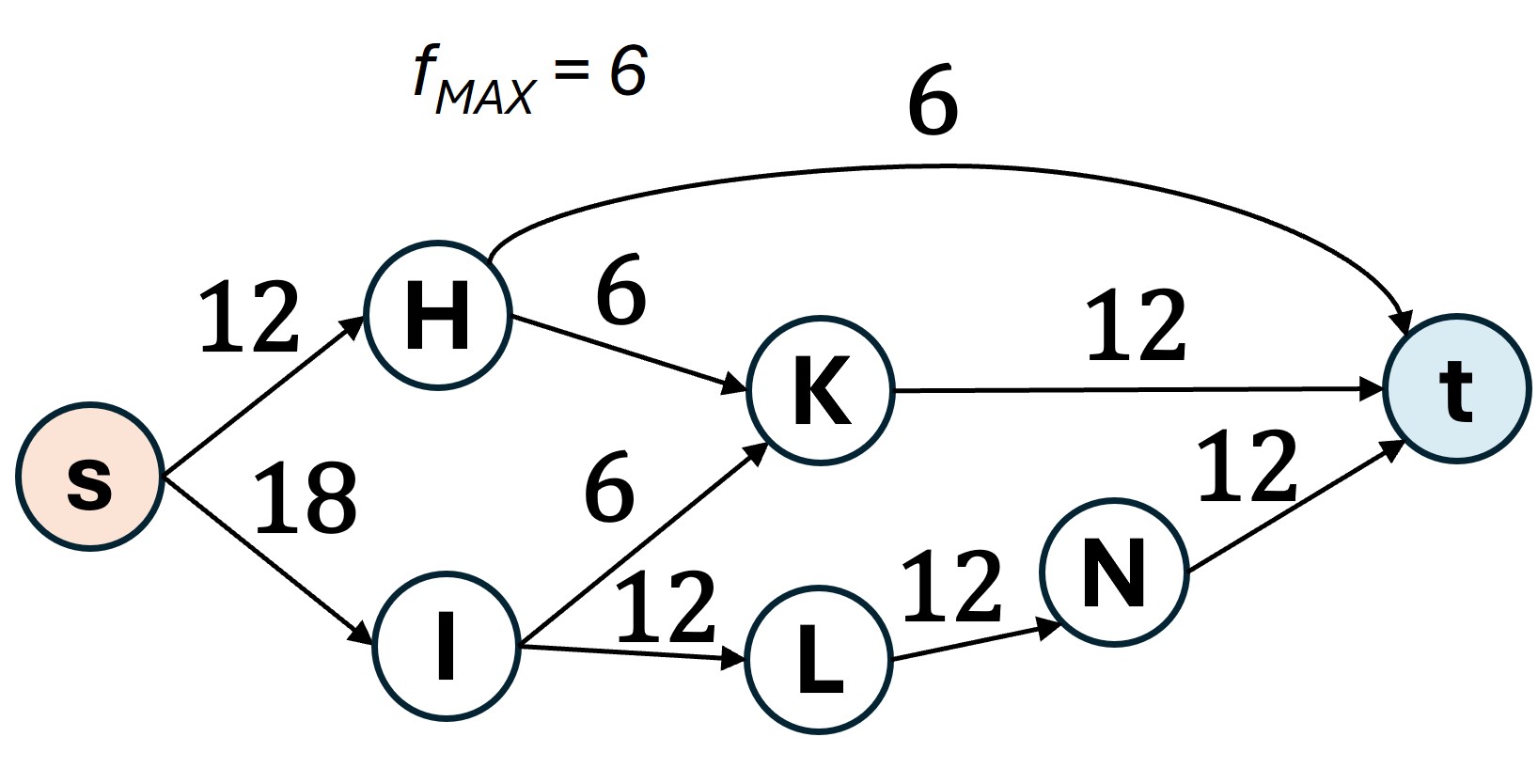}
\caption{Augmented graph $G'$ with initial flow $f$ along each edge labeled}
\label{fig:MaxTopCutFlow}
\end{subfigure}

\vspace{0.4em}

\begin{subfigure}[t]{0.48\columnwidth}
\centering
\includegraphics[width=\linewidth]{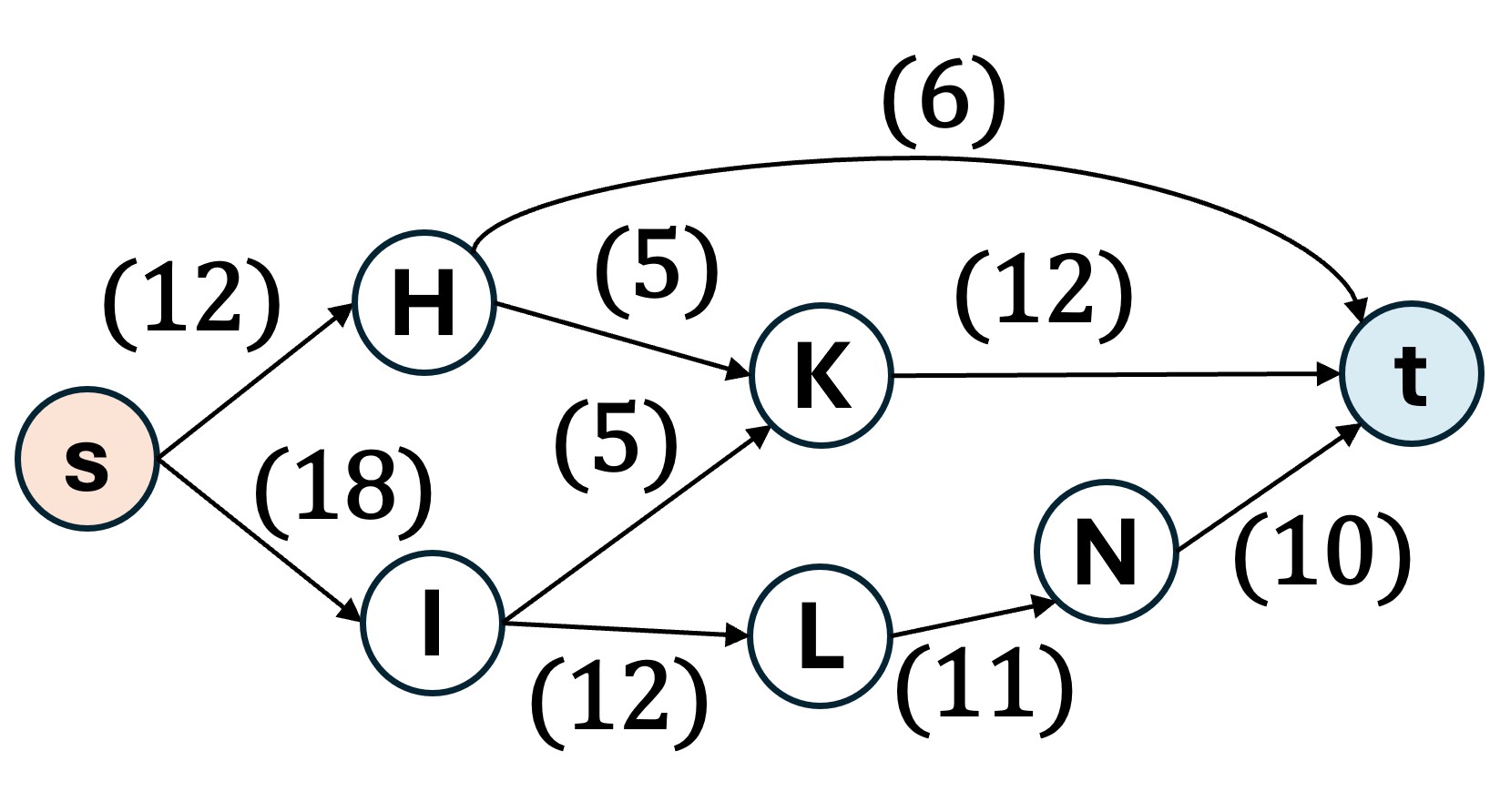}
\caption{$s$-$t$ graph $G^+$ with capacities assigned as flow $f$ along edge in $G'$ - edge weight in $G$}
\label{fig:MaxTopCutST}
\end{subfigure}\hfill
\begin{subfigure}[t]{0.48\columnwidth}
\centering
\includegraphics[width=\linewidth]{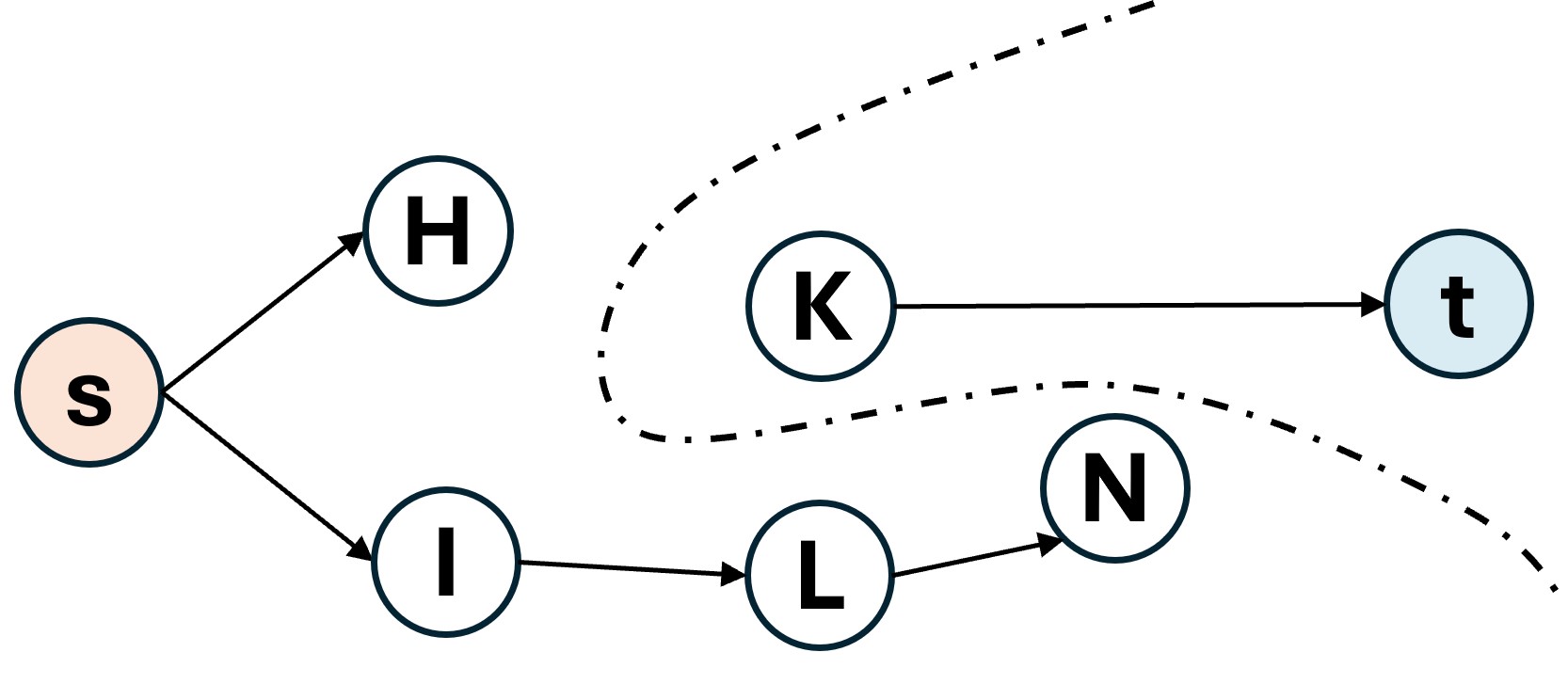}
\caption{Resulting min-cut giving the optimal MaxTopCut of (a)}
\label{fig:MaxTopCutresidual}
\end{subfigure}

\caption{Maximum topological cut via a minimum $s$-$t$ cut reduction.}
\label{fig:BCIFormulation}

\vspace{-0.8em}
\end{figure}
We optimally solve for BCI locations using the MaxTopCut formulation from~\cite{MaxTopCut}, which reduces the problem to a minimum $s$-$t$ cut. We additionally incorporate BCI-specific constraints: (i) all edges in $C$ must belong to the cut $(S,T)$, (ii) primary inputs (row $r-1$) must lie in $S$, and (iii) primary outputs (row $D+r$) must lie in $T$.

To enforce these constraints, we construct an augmented graph $G'$ by merging all primary inputs into a super-source $s$ and all primary outputs into a super-sink $t$ (see Fig.~\ref{fig:MaxTopCutG}). The edges in $C$ are removed, and the transitive fanouts of their sink nodes are merged into $t$, ensuring that these edges are necessarily cut. Any resulting direct edges between $s$ and $t$ are removed. Notably, collapsing multiple nodes into $t$ causes some of the resulting augmented edges in $G'$ to correspond to multiple edges of $G$.

In $G'$, we assign edge weights of $1$ for each violating edge  $e$ of $G$ contained in the augmenting edge $e'$. Following~\cite{MaxTopCut}, we define

\begin{equation}
f_{\max} = 1 + \sum_{e'\in E'} w(e'),
\end{equation}
and construct a flow $f$ such that $f(i,j) \ge f_{\max}$ for all edges (Fig.~\ref{fig:MaxTopCutFlow}). A new graph $G^+$ (Fig.~\ref{fig:MaxTopCutST}) is 
\begin{equation}
c^+_{ij} = f(i,j) - w(e').
\end{equation}

A minimum $s$-$t$ cut in $G^+$ yields the maximum topological cut of $G'$~\cite{MaxTopCut} (Fig.~\ref{fig:MaxTopCutresidual}). The resulting solution is obtained from reachability in the residual graph, where the set of nodes reachable from $s$ defines $S$, and edges from $S$ to $T = V \setminus S$ constitute the cut. Optimal edge selection for BCI can therefore be solved via a maximum flow computation in polynomial time~\cite{MaxTopCut}, with complexity $O(|V||E|)$ using Orlin's algorithm~\cite{OrlinFlow}.

\begin{algorithm}[t]
\caption{BCI$(G,C, A, r,D)$}
\label{alg:bci}
\begin{algorithmic}[1]
\Statex \textbf{Input:}
\Statex $G=(V,E)$: DAG representation of the circuit
\Statex $C$: set of violating edges from row $r-1$
\Statex $A$: set of adjacent non-violating edges
\Statex $r$: source logic level under consideration
\Statex $D$: maximum depth explored for cut insertion
\Statex \textbf{Output:}
\Statex $\mathrm{\textit{edge\_list}}$: set of edges selected for buffer insertion

\State Assign edge weights:
\Statex \hspace{1em} $w(e) \gets 1$ if $e$ is violating, else $0$

\State Construct augmented graph $G' = (V', E')$:
\Statex \hspace{1em} $s \gets \{ \mathrm{src}(e) \mid e \in A \cup C \}$
\Statex \hspace{1em} $t \gets \{\mathrm{TFO}(\mathrm{sink}(e)\mid e \in C\} \cup \{ v \mid \mathrm{level}(v)=r+D\}$

\State Compute $f_{\max} \gets 1 + \sum_{e' \in E'} w(e')$
\State Construct feasible flow $f$ such that $f(e') \ge f_{\max}$ for all $e'$

\State Build graph $G^+$ with capacities:
\Statex \hspace{1em} $c(e) \gets f(e) - w(e')$

\State Compute minimum $s$-$t$ cut $(S,T)$ on $G^+$

\State $\mathrm{\textit{edge\_list}} \gets \{ (u,v) \in E' \mid u \in S, v \in T \}$

\State \Return $\mathrm{\textit{edge\_list}}$

\end{algorithmic}
\end{algorithm}

Algorithm~\ref{alg:bci} summarizes this procedure and returns the maximum topological cut containing all edges in $C$, which is then used by Algorithm~\ref{alg:retime_engine} to insert buffers.
\setcounter{table}{1}
\begin{table*}[tb]
\centering
\caption{Comparison of placement overhead to state-of-the-art methods in terms of depth increase, inserted JJs, latency, and runtime.}
\label{tab:full_comparison}

\setlength{\tabcolsep}{4pt}

\begin{adjustbox}{max width=\textwidth}
\begin{tabular}{l|ccccc|cc|cc|cc}
\hline
\multirow{2}{*}{\textbf{Benchmark}} &
\multicolumn{5}{c|}{\textbf{Depth Increase (BRI/BCI)}} &
\multicolumn{2}{c|}{\textbf{Inserted JJs}} &
\multicolumn{2}{c|}{\textbf{Latency (ps)}} &
\multicolumn{2}{c}{\textbf{Runtime (s)}} \\

&
\textbf{GORDIAN~\cite{GORDIAN}} &
\textbf{TAAS~\cite{TAAS}} &
\textbf{SuperFlow~\cite{superflow}} &
\textbf{DLPlace~\cite{DLPlace}} &
\textbf{RIVERPlace} &
\textbf{SuperFlow~\cite{superflow}} &
\textbf{RIVERPlace} &
\textbf{DLPlace~\cite{DLPlace}} &
\textbf{RIVERPlace} &
\textbf{SuperFlow~\cite{superflow}} &
\textbf{RIVERPlace} \\
\hline

\texttt{adder8}   & 24 & 24 & 16 & \textbf{0}  & 1  & 1,210  & \textbf{70}   & 594  & \textbf{402}  & 12.1  & \textbf{4.3} \\
\texttt{apc32}    & 26 & 26 & 26 & \textbf{0}  & 1  & 1,294  & \textbf{82}   & 625  & \textbf{496}  & 13.8  & \textbf{4.9} \\
\texttt{apc128}   &117 &110 & 67 & 24 & \textbf{4}  & 8,812  & \textbf{592}  &2401 & \textbf{2091} &374.8 & \textbf{34.8} \\
\texttt{c432}     & 46 & 45 & 29 & 8  & \textbf{2}  & 2,786  & \textbf{170}  &1310 & \textbf{800}  &162.5 & \textbf{14.5} \\
\texttt{c499}     & 62 & 62 & 59 &40  & \textbf{7}  & 14,104  & \textbf{1,262} &2926 & \textbf{1459} &113.4 & \textbf{33.8} \\
\texttt{c1355}    & 58 & 58 & 56 &32  & \textbf{13} & 16,008  & \textbf{2,286} &2631 & \textbf{1866} & 50.1 & \textbf{45.4} \\
\texttt{c1908}    & 67 & 66 & 68 &32  & \textbf{6}  & 10,692 & \textbf{888}  &2681 & \textbf{1109} &517.5 & \textbf{37.4} \\
\texttt{decoder}  & 34 & 33 & 43 &13  & \textbf{4}  & 5,686  & \textbf{560}  &1092 & \textbf{643}  &690.9 & \textbf{12.3} \\
\texttt{sorter32} & 29 & 29 & 29 &27  & \textbf{12} & 4,980  & \textbf{1,522} &2046 & \textbf{1167} &353.3 & \textbf{33.3} \\
\hline

\textbf{Avg. Ratio}
&15.3&15.0&12.3&3.2&\textbf{1.0}
&12.0&\textbf{1.0}
&1.6&\textbf{1.0}
&12.5&\textbf{1.0}\\
\hline
\end{tabular}
\end{adjustbox}
\end{table*}

\setcounter{table}{0}
\begin{table}[htbp]
\centering
\caption{AQFP cell library used for all designs.}
\label{tab:cell_drive}
\begin{tabular}{lc}
\toprule
\textbf{Cell Type} &
\textbf{Max. Interconnect ($\mu$m)}\\
\midrule
\texttt{Buffers}     &  800\\
\texttt{Logic Gates} &  300\\
\texttt{Splitters} &  200\\

\bottomrule
\end{tabular}
\end{table}

\section{Results}

We evaluate RIVERPlace on both the published AQFP benchmark suite used by prior placement works~\cite{GORDIAN,TAAS,superflow,DLPlace} and substantially larger open-source AQFP netlists. Table~\ref{tab:cell_drive} summarizes the interconnect limits used for each cell type. Post-routing timing closure is validated using~\cite{qPRO}. Unless otherwise stated, reported averages are normalized to RIVERPlace and averaged across all completed benchmarks. All RIVERPlace experiments were conducted on an EPYC 7763 Milan CPU.

\subsection{Comparison to the SOTA}\label{sec:SOTA}

We compare RIVERPlace against the state-of-the-art AQFP placement approaches GORDIAN (DATE'21)~\cite{GORDIAN}, TAAS (DAC'22)~\cite{TAAS}, DLPlace (ICCAD'23)~\cite{DLPlace}, and SuperFlow (DATE'24)~\cite{superflow} using the common AQFP benchmark suite from~\cite{Chen2019}.

Table~\ref{tab:full_comparison} compares the placement overhead, JJ count, latency, and runtime. With the exception of inserting one additional row on the small benchmarks \texttt{adder8} and \texttt{apc32} compared to DLPlace, RIVERPlace outperforms all prior methods across every other reported metric.

More specifically, while DLPlace represents the previous state of the art in placement quality, RIVERPlace achieves more than a 3$\times$ reduction in row insertion overhead and a 38\% reduction in average 
latency.\footnote{DLPlace does not report JJ count or runtime, nor does it demonstrate scalability to larger benchmarks, as it repeats the complete global and detailed placement flow after each buffer insertion.}
In contrast, SuperFlow represents the current state of the art in runtime scalability, while providing only moderate improvements in placement quality over GORDIAN and TAAS. Compared to SuperFlow, RIVERPlace achieves, on average, more than an order-of-magnitude reduction in inserted rows, inserted JJ count, and runtime.\footnote{Note that SuperFlow reports runtimes on a different processor; the runtime comparison should be interpreted qualitatively.}

\subsection{Scalability to Larger Designs}

A key limitation of prior AQFP placement literature is that comparisons between placement algorithms have largely been restricted to the relatively small benchmark circuits reported in Table~\ref{tab:full_comparison}. 

\setcounter{table}{2}
\begin{table}[htbp]
\centering
\caption{Placement results of GORDIAN and RIVERPlace on benchmark circuits on the full AQFP benchmark suite. (Area in $mm^2$ and runtime in seconds).}
\label{tab:riverplace_placement_results}

\setlength{\tabcolsep}{3pt}
\renewcommand{\arraystretch}{1.08}
\begin{adjustbox}{max width = \linewidth}

\begin{tabular}{lccccccccc}
\toprule
\multirow{2}{*}{\textbf{Benchmarks}} &
\multicolumn{2}{c}{\textbf{Input Netlist~\cite{LayoutAwareAQFPGitHub}}} &
\multicolumn{3}{c}{\textbf{GORDIAN~\cite{GORDIAN}}} &
\multicolumn{4}{c}{\textbf{RIVERPlace}} \\
\cmidrule(lr){2-3}\cmidrule(lr){4-6}\cmidrule(lr){7-10}
& \textbf{Cells} & \textbf{Depth} &
\textbf{BRI} & \textbf{JJs} & \textbf{Area} &
\textbf{BCI} & \textbf{JJs} & \textbf{Area} & \textbf{Runtime} \\
\midrule
\texttt{adder1}     & 25     & 8   & 0   & 78      & 0.11   & 0  & 78      & 0.09  & 0.5 \\
\texttt{adder8}     & 627    & 35  & 2   & 1,642   & 2.22   & 1  & 1,586   & 1.42  & 4.9 \\
\texttt{mult8}      & 3,036  & 74  & 54  & 13,194  & 19.55  & 8  & 8,794   & 9.64  & 110 \\
\texttt{apc16}      & 131    & 18  & 2   & 416     & 0.82   & 0  & 378     & 0.53  & 1.3 \\
\texttt{apc32}      & 307    & 24  & 4   & 1,058   & 2.06   & 1  & 976     & 1.41  & 3.2 \\
\texttt{apc64}      & 696    & 31  & 15  & 3,128   & 6.85   & 2  & 2,312   & 3.63  & 10.4 \\
\texttt{apc128}     & 1,533  & 39  & 41  & 9,516   & 23.78  & 4  & 5,218   & 9.27  & 27.6 \\
\texttt{c17}        & 19     & 5   & 0   & 62      & 0.11   & 0  & 62      & 0.09  & 0.4 \\
\texttt{c432}       & 1,040  & 39  & 13  & 3,434   & 5.21   & 3  & 2,782   & 3.19  & 15.2 \\
\texttt{c499}       & 1,770  & 32  & 48  & 12,098  & 20.79  & 8  & 6,534   & 8.73  & 44.2 \\
\texttt{c880}       & 2,152  & 40  & 37  & 10,192  & 13.74  & 6  & 6,424   & 6.89  & 49.8 \\
\texttt{c1355}      & 1,994  & 36  & 48  & 12,354  & 18.53  & 7  & 6,772   & 7.60  & 42.5 \\
\texttt{c1908}      & 1,714  & 37  & 33  & 8,230   & 12.00  & 6  & 5,472   & 5.92  & 38.1 \\
\texttt{c2670}      & 3,147  & 29  & 55  & 20,248  & 28.17  & 7  & 9,816   & 9.77  & 63.4 \\
\texttt{c3540}      & 4,557  & 54  & 134 & 40,222  & 60.92  & 40 & 23,428  & 26.54 & 466.2 \\
\texttt{sorter32}   & 1,054  & 32  & 31  & 7,032   & 10.62  & 13 & 5,660   & 6.45  & 39.2 \\
\texttt{sorter48}   & 2,195  & 40  & 81  & 20,934  & 32.95  & 20 & 12,076  & 13.46 & 131.0 \\
\texttt{c5315}      & 8,228  & 41  & 258 & 142,468 & 192.34 & 65 & 55,476  & 60.11 & 1,674.5 \\
\texttt{c6288}      & 11,039 & 181 & 302 & 80,344  & 118.97 & 53 & 40,636  & 44.74 & 1,994.4 \\
\texttt{c7552}      & 11,664 & 56  & 255 & 150,632 & 212.51 & 81 & 65,612  & 71.15 & 2,922.1 \\
\texttt{alu32}      & 18,586 & 171 & 339 & 138,448 & 212.41 & 62 & 64,248  & 78.77 & 4,709.13 \\
\midrule
\multicolumn{3}{c}{\textbf{Average Ratio}} &
\textbf{5.29} & \textbf{1.6} & \textbf{2.10} &
\textbf{1} & \textbf{1} & \textbf{1} & \textbf{--} \\
\bottomrule
\end{tabular}
    
\end{adjustbox}
\end{table}
Therefore, to demonstrate scalability, we use a recently open-sourced AQFP benchmark set that contains much larger circuits~\cite{LayoutAwareAQFPGitHub}. Because the only prior SOTA codebase available to us is GORDIAN, we rerun GORDIAN directly using identical netlists, cell libraries, and experimental settings as a baseline.

These benchmark netlists have been optimized specifically for AQFP placement~\cite{LayoutAwareAQFP}, resulting in substantially lower placement overheads for all
methods. 

Nevertheless, RIVERPlace still delivers more than a 2$\times$ reduction in total circuit area, as reported in Table~\ref{tab:riverplace_placement_results}.

Current AQFP fabrication technology supports approximately 20,000 Josephson junctions~\cite{MANA,JJ20K}. Using GORDIAN, circuits such as \texttt{c5315} exceed this limit by more than a factor of seven, whereas RIVERPlace reduces the overhead to approximately three times the current fabrication limit. These results demonstrate that globally optimizing pipelining locations through BCI is substantially more effective than uniform buffer row insertion for large AQFP designs and provides a scalable path as fabrication capabilities continue to improve.

\subsection{Results after Routing and Timing Closure}

We integrate RIVERPlace into the complete RTL-to-GDSII qPALACE flow~\cite{qPALACE}, which includes AQFP routing and post-routing timing closure~\cite{qPRO}. The results in Table~\ref{tab:riverplace_routed_results} represent the first post-routing, timing-closed results for the larger benchmark suite reporting both area and performance. Notably, prior AQFP placement works~\cite{DLPlace,TAAS,GORDIAN,superflow} generally do not report post-placement or post-routing area, even for substantially smaller circuits.

To ensure that placement optimization accurately reflects post-routing interconnect lengths, RIVERPlace accounts for routing density during legalization using the interval-density track estimation model described in Section~III-E. The estimated routing demand is used to dynamically adjust channel widths and row spacing, allowing placement to be optimized with routing requirements incorporated directly into interconnect constraints. As demonstrated by the post-routing results in Table~\ref{tab:riverplace_routed_results}, the resulting placements remain routable and satisfy interconnect constraints while preserving RIVERPlace's substantial area advantage. Figure~\ref{fig:apc128} illustrates the resulting post-routing, timing-closed layout of the \texttt{apc128} benchmark.

\begin{table}[htbp]
\centering
\caption{Routed results RIVERPlace on benchmark circuits.}
\label{tab:riverplace_routed_results}

\setlength{\tabcolsep}{3pt}
\renewcommand{\arraystretch}{1.08}
\begin{tabular}{lccc}
\toprule
\multirow{2}{*}{\textbf{Benchmarks}} &
\multicolumn{3}{c}{\textbf{RIVERPlace-Routed Results}} \\
\cmidrule(lr){2-4}
& \textbf{Latency (ps)} & \textbf{Performance (GHz)} & \textbf{Area ($mm^2$)} \\
\midrule
\texttt{adder1}     & 71    & 5.0 & 0.27 \\
\texttt{adder8}     & 508   & 5.0 & 2.06 \\
\texttt{mult8}      & 1,925 & 5.0 & 11.03 \\
\texttt{apc16}      & 226   & 5.0 & 0.86 \\
\texttt{apc32}      & 381   & 5.0 & 1.86 \\
\texttt{apc64}      & 657   & 5.0 & 4.15 \\
\texttt{apc128}     & 1,390 & 5.0 & 9.99 \\
\texttt{c17}        & 42    & 5.0 & 0.22 \\
\texttt{c432}       & 709   & 5.0 & 3.90 \\
\texttt{c499}       & 1,335 & 5.0 & 9.44 \\
\texttt{c880}       & 1,249 & 5.0 & 7.68 \\
\texttt{c1355}      & 1,282 & 5.0 & 8.36 \\
\texttt{c1908}      & 1,109 & 5.0 & 6.64 \\
\texttt{c2670}      & 1,489 & 4.2 & 10.36 \\
\texttt{c3540}      & 4,040 & 4.0 & 28.14 \\
\texttt{sorter32}   & 1,190 & 5.0 & 7.21 \\
\texttt{sorter48}   & 2,169 & 5.0 & 14.43 \\
\texttt{c5315}      & 8,805 & 2.3 & 61.85 \\
\texttt{c6288}      & 7,487 & 5.0 & 48.65 \\
\texttt{c7552}      & 10,578 & 2.3 & 73.20 \\
\texttt{alu32}      & 12,030 & 3.5 & 82.58 \\
\bottomrule
\end{tabular}
\end{table}

\begin{figure}[htbp]
    \centering
    \begin{minipage}{1\linewidth}
        \centering
        \includegraphics[width=0.8\linewidth]{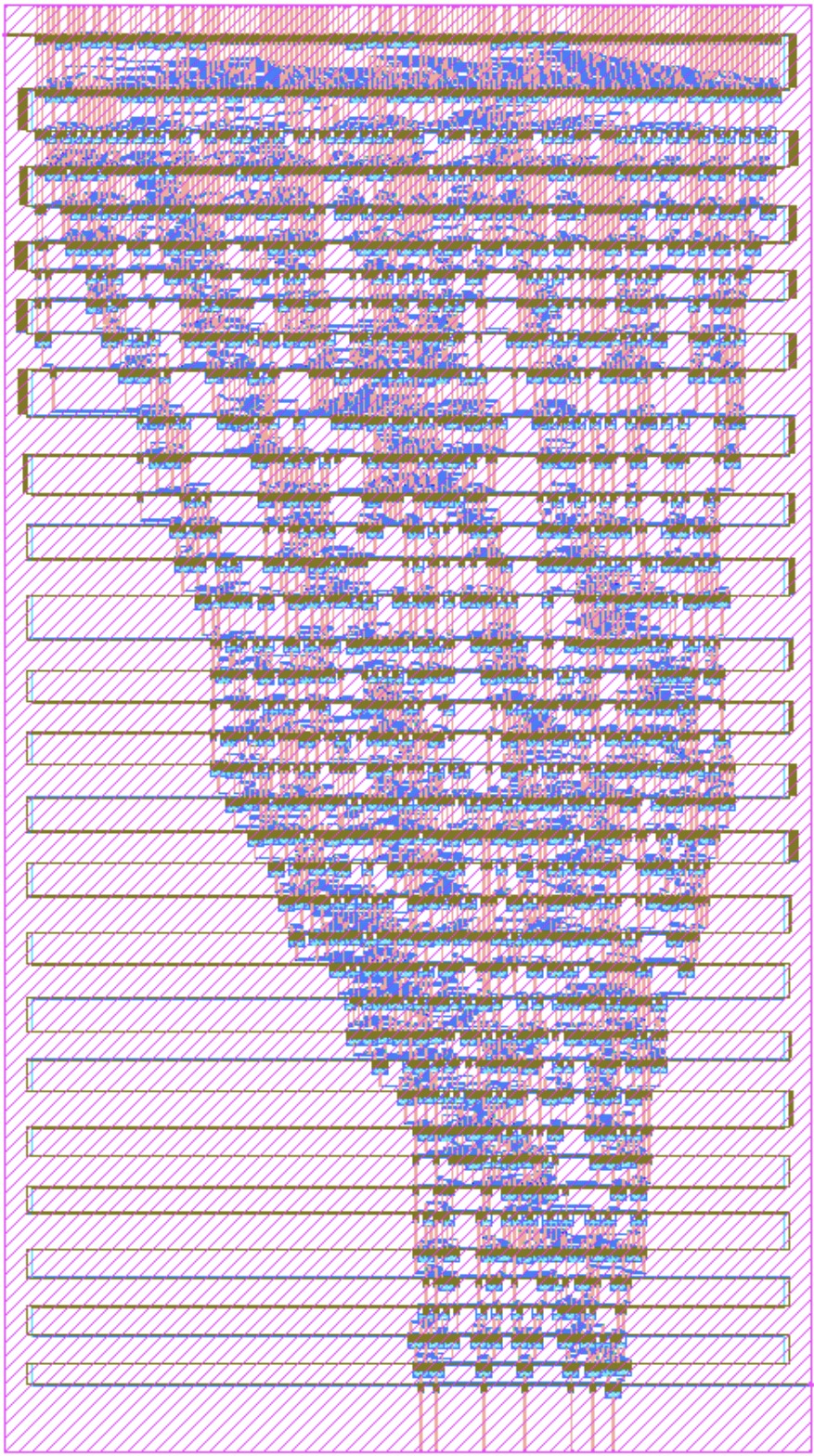}
    \end{minipage}

    \caption{
    Placed and routed design of \texttt{apc128} using RIVERPlace integrated into qPALACE.
    }
    \label{fig:apc128}
\end{figure}

\subsection{Ablations and Analysis}

To better understand the contributions and robustness of the proposed framework, we perform a series of ablation studies and sensitivity analyses.

\subsubsection{Component analysis}

RIVERPlace combines several algorithmic innovations into a unified placement framework. To evaluate the contribution of each component, Fig.~\ref{fig:ablation} presents an ablation study comparing three variants of the proposed flow, with all results normalized to GORDIAN. LP+BRI achieves an average normalized depth of 0.79, LP+Retiming+BRI reduces this to 0.72, and LP+BCI further improves it to 0.65. The complete RIVERPlace flow achieves the lowest normalized depth of 0.64.

These results show that the LP-based incremental placement formulation provides the foundation of RIVERPlace, resolving the majority of interconnect violations through placement optimization alone. Placement-aware retiming further reduces circuit depth by reusing existing buffers without introducing additional pipeline stages, while BCI delivers the largest additional improvement by globally selecting pipelining locations when depth increases become unavoidable.

\begin{figure}[htbp]
    \centering
    \begin{minipage}{1\linewidth}
        \centering
        \includegraphics[width=\linewidth]{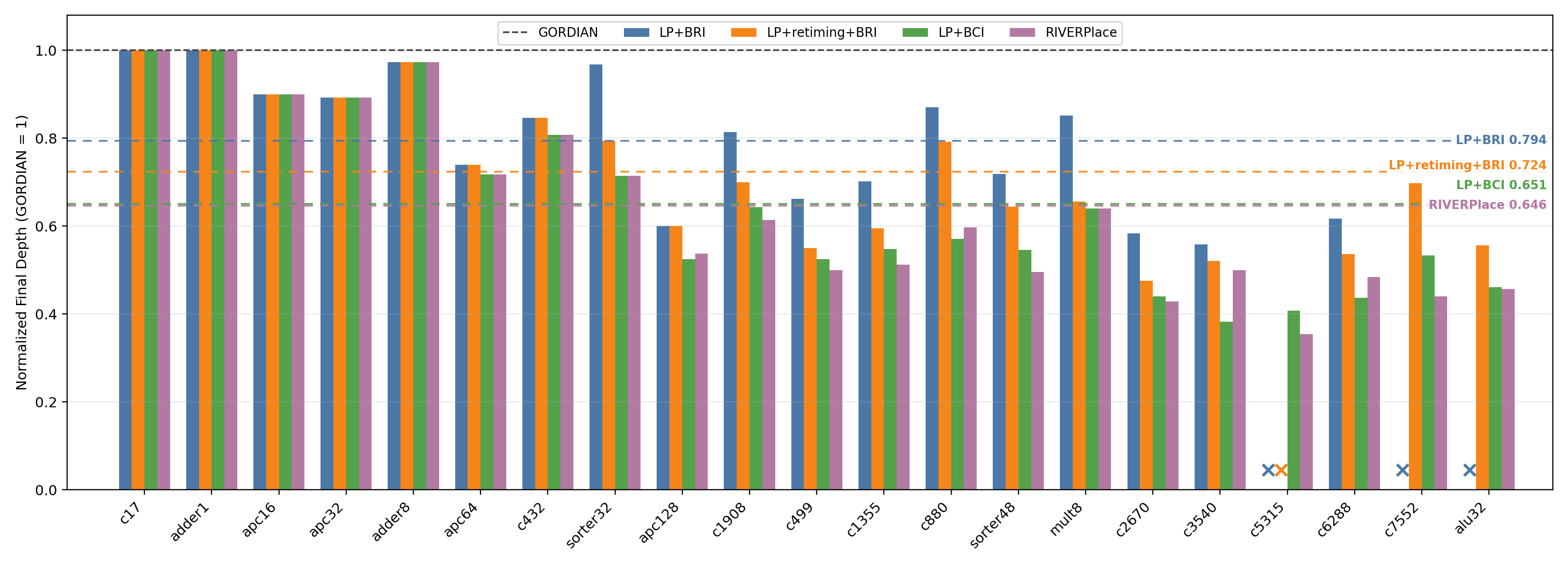}
    \end{minipage}
    \caption{
    Ablation study illustrating the impact of each RIVERPlace component (LP, retiming, and BCI) relative to GORDIAN. 
    }

    \label{fig:ablation}
\end{figure}
\subsubsection{LP weighting sensitivity}

To evaluate the robustness of the LP formulation, we perform a sensitivity study on the movement weight $\lambda_m$, fixing $\lambda_s=1$ while sweeping $\lambda_m$ from 0 to 1. The default setting is $\lambda_m=0.1$.  As shown in Fig.~\ref{fig:lambda_sensitivity}, the required BCI count is largely insensitive to $\lambda_m$ in the vicinity of the default setting for \texttt{c3540}, with similar trends observed across the benchmark suite.  As $\lambda_m$ increases, cell movement becomes increasingly constrained, reducing the LP's ability to resolve violations through placement refinement and consequently requiring more BCI. These results indicate that the LP framework is robust to moderate changes in objective weighting provided that slack reduction remains the dominant objective ($\lambda_m \ll \lambda_s$).

\begin{figure}[htbp]
    \centering
        \includegraphics[width=\linewidth]{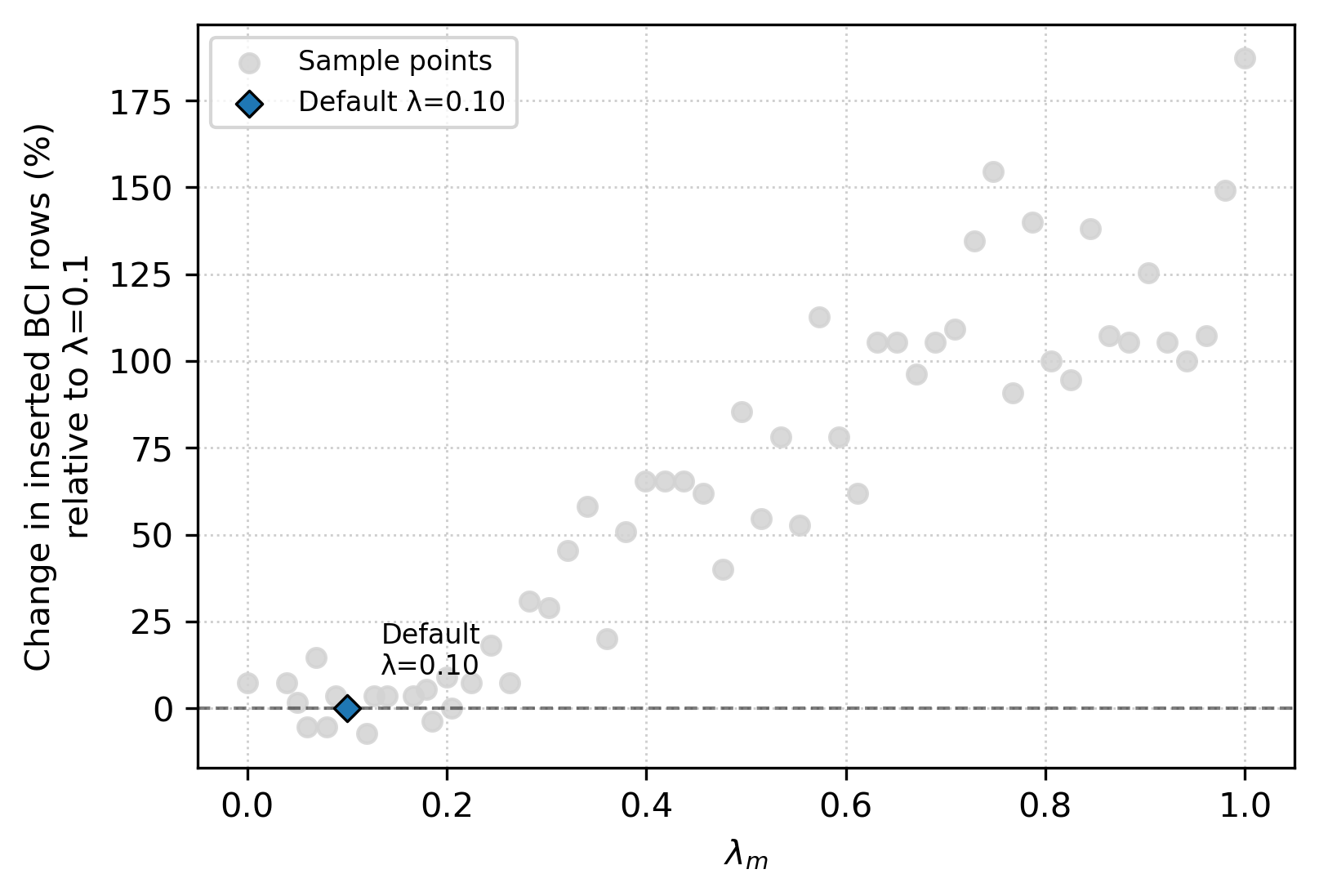}
    \caption{
    BCI overhead is insensitive to LP weights for $\lambda_m \ll \lambda_s$
    }
    \label{fig:lambda_sensitivity}
\end{figure}

\subsubsection{Retiming engine iteration analysis}

Because RIVERPlace's retiming engine modifies downstream logic levels, each 
invocation resolves violations iteratively, processing rows in topological 
order and freezing previously legalized rows. This strategy confines structural changes to downstream logic, producing a monotonic progression toward a legal placement, as illustrated for \texttt{alu32} in Fig.~\ref{fig:alu32_it}.

\begin{figure}[htbp]
    \centering
        \centering
        \includegraphics[width=\linewidth]{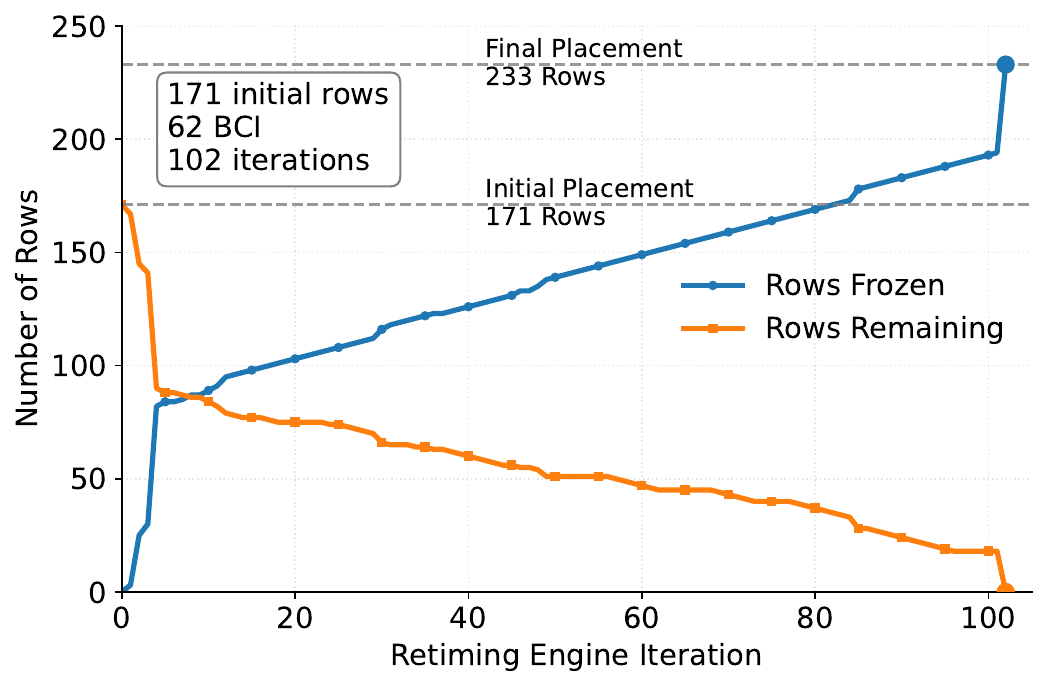}
    \caption{
    Placed and remaining rows during placement of \texttt{alu32}. Increases in total row count correspond to BCI introducing additional pipeline stages to resolve interconnect violations.
    }
    \label{fig:alu32_it}
\end{figure}

To quantify convergence, we count how many times each row invokes the retiming engine across the entire benchmark suite~\cite{LayoutAwareAQFPGitHub}. As shown in Fig.~\ref{fig:global_iteration_analysis}, 97.4\% of all rows require at most a single retiming invocation because violations were either eliminated directly by the retiming call or through upstream BCI decisions and LP-based placement refinement. Across the entire benchmark suite (1,409 rows), only 37 rows require multiple iterations, with the most demanding row converging after five passes of retiming existing buffers without increasing circuit depth. These results demonstrate that the LP formulation and retiming engine cooperate effectively to produce rapid and stable convergence.

\begin{figure}[htbp]
    \centering
        \centering
        \includegraphics[width=\linewidth]{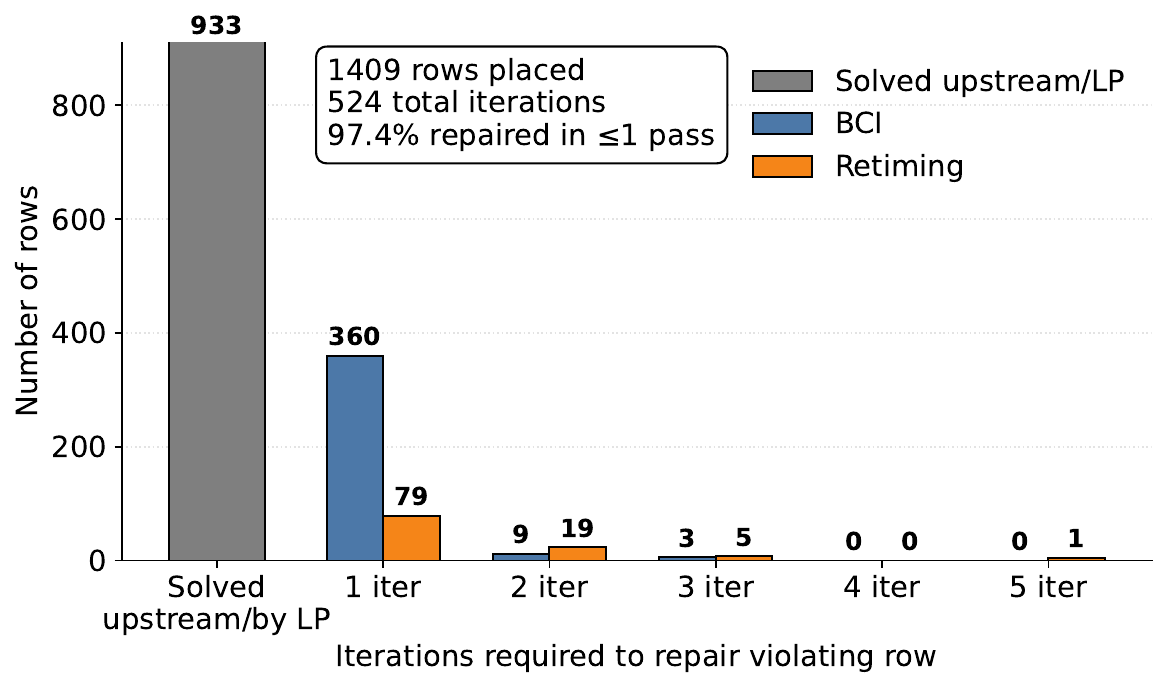}
    \caption{
    Distribution of the number of retiming iterations required per row across the benchmark suite.
    }
    \label{fig:global_iteration_analysis}
\end{figure}
Together, these results demonstrate that the quality improvements of RIVERPlace arise from the complementary interaction between incremental placement, placement-aware retiming, and globally optimized pipelining.
\section{Conclusions}

In this work, we presented RIVERPlace, an integrated retiming and incremental placement framework for resolving interconnect violations in AQFP circuits. Central to our approach is Buffer Cut Insertion (BCI), which casts violation resolution as a constrained global edge-selection problem. We show that, under AQFP path-balancing constraints, this problem reduces to a maximum topological cut, enabling an exact polynomial-time solution. Beyond BCI, RIVERPlace interleaves placement-aware retiming, local reordering, LP-based coordinate refinement, and incremental placement to resolve interconnect violations while minimizing disruption to previously optimized regions.

Experimental results demonstrate RIVERPlace's approach to efficient violation repair substantially reduces AQFP placement overhead. Compared to prior state-of-the-art AQFP placement methods, RIVERPlace achieves over an order-of-magnitude reduction in inserted JJs and runtime, a threefold reduction in placement-induced depth overhead, and a 38\% reduction in latency. Furthermore, RIVERPlace enables the first post-routing, timing-closed implementations of the complete open-source AQFP benchmark suite, including larger circuits like \texttt{alu32}.

More broadly, this work highlights the importance of incorporating pipelining and retiming directly into placement for superconducting logic. By formulating pipelining decisions as a constrained global edge-selection problem subject to topological constraints, RIVERPlace introduces a new algorithmic perspective on AQFP physical design. Future work includes integrating pipelining earlier into the synthesis flow and jointly optimizing logic restructuring, placement, and routing.

\section{Acknowledgments}
We used OpenAI ChatGPT and Google Gemini only for grammar correction and language polishing. These tools did not generate research ideas, technical content, results, figures,
or references.

This work has been supported by ARL DEVCOM under the FSDL: ColdPhase project, grant number W911NF-24-1-0317.

\bibliographystyle{IEEEtran}
\bibliography{bibliography}

\end{document}